\documentclass[journal,twoside,web]{ieeecolor}

\usepackage{generic}
\usepackage{cite}
\usepackage{amsmath,amssymb,amsfonts}
 \usepackage{algorithm}
\usepackage{algorithmicx}
\usepackage{algpseudocode}
\usepackage{graphicx}
\usepackage{subcaption}
\usepackage{textcomp}
\usepackage{bbm}
\usepackage{array}
\usepackage{stfloats}
\usepackage{times}
\usepackage{enumerate}
\usepackage{mathtools}
\usepackage{dsfont}
\usepackage{amssymb}
\usepackage{multirow}
\usepackage{gensymb}
\newtheorem{theorem}{\bf Theorem}

\newtheorem{problem}{\bf Problem}
\newtheorem{remark}{\bf Remark}
\newtheorem{definition}{\bf Definition}

\newtheorem{lemma}{\bf Lemma}
\newtheorem{assumption}{\bf Assumption}

\usepackage{algpseudocode}
\DeclareMathOperator*{\argmin}{arg\,min}

\def\QED{~\rule[-1pt]{5pt}{5pt}\par\medskip}
\begin{document}

\title{Strong Duality and Dual Ascent Approach to Continuous-Time Chance-Constrained Stochastic Optimal Control}
\author{Apurva Patil, Alfredo Duarte, Fabrizio Bisetti, and Takashi Tanaka
% \thanks{This work is supported by NSF Award 1944318..}
\thanks{A. Patil is with the Walker Department of Mechanical Engineering, University of Texas at Austin, Austin, TX 78712 USA (e-mail: apurvapatil@utexas.edu). A. Duarte, F. Bisetti, and T. Tanaka are with the Department of Aerospace Engineering and Engineering Mechanics at the University of Texas at Austin, Austin, TX 78712 USA (e-mail: aduarteg@utexas.edu, fbisetti@utexas.edu, and ttanaka@utexas.edu). }
}

\maketitle

\begin{abstract}
The paper addresses a continuous-time continuous-space chance-constrained stochastic optimal control (SOC) problem where the probability of failure to satisfy given state constraints is explicitly bounded. We leverage the notion of exit time from continuous-time stochastic calculus to formulate a chance-constrained SOC problem. Without any conservative approximation, the chance constraint is transformed into an expectation of an indicator function which can be incorporated into the cost function by considering a dual formulation. We then express the dual function in terms of the solution to a Hamilton-Jacobi-Bellman partial differential equation parameterized by the dual variable. Under a certain assumption on the system dynamics and cost function, it is shown that a strong duality holds between the primal chance-constrained problem and its dual. The Path integral approach is utilized to numerically solve the dual problem via gradient ascent using open-loop samples of system trajectories. We present simulation studies on chance-constrained motion planning for spatial navigation of mobile robots and the solution of the path integral approach is compared with that of the finite difference method.
\end{abstract}

\begin{IEEEkeywords}
Chance-constrained stochastic optimal control, Exit time, Strong duality, Dual ascent, Numerical algorithms.
%visit \underline{http://www.ieee.org/organizations/pubs/ani\_prod/keywrd98.txt}
\end{IEEEkeywords}

\section{Introduction}
\subsection{Motivation}
\IEEEPARstart{I}{n} safety-critical missions, quantitative characterization of system uncertainties is of critical importance as it impacts the overall safety of the system operation. System uncertainties arise due to unmodeled dynamics, unknown system parameters, and external disturbances. Subsequently, the control policies should be designed to accommodate such uncertainties in order to achieve user-defined safety requirements. Robust control is a popular paradigm to guarantee safety against set-valued uncertainties \cite{kothare1996robust, kuwata2007robust}. Generally, a robust control approach aims to synthesize a policy that optimizes the worst-case performance, commonly known as the minimax policy. While such a strategy is suitable for applications where safety is an absolute requirement, the computation of the exact minimax policy is often intractable, which necessitates a sequential outer-approximation of the uncertain sets. This often results in overly conservative solutions \cite{calafiore2006scenario}. Also, robust control is difficult to apply when the uncertainty is modeled probabilistically using random variables with unbounded support (e.g., Gaussian distributions).\par

\textit{Chance-constrained} stochastic optimal control (SOC) is an alternative paradigm for policy synthesis under uncertainty \cite{blackmore2011chance, oguri2019convex}. Unlike robust control, this approach aims to optimize the system performance by accepting a user-specified threshold for the probability of failure (e.g., collisions with obstacles). Notably, the acceptance of the possibility of failure is often effective to reduce the conservatism of the controller even if the introduced probability of failure is practically negligible \cite{vidyasagar2001randomized}. Consequently, chance-constrained SOC is also widely used as a framework for policy synthesis for systems in safety-critical missions.\par

\subsection{Literature Review}
In this paper, we consider a continuous-time continuous-space chance-constrained SOC problem. A central challenge in this problem stems from the fact that the continuous-time end-to-end probability of failure is generally challenging to evaluate and optimize against. A common course of action found in the literature is to look for a tractable approximation of the chance-constrained SOC problem in order to use existing tools from optimization.  In \cite{blackmore2011chance} a discrete-time chance-constrained SOC problem is converted to a disjunctive convex program by approximating the chance constraint using Boole’s inequality. The disjunctive convex program is then solved using branch-and-bound techniques. The work presented in \cite{ono2015chance} solves the discrete-time chance-constrained problem by formulating its dual. However, in this work, the joint chance constraint is conservatively approximated using Boole's inequality. Hence, the duality gap is nonzero and the obtained solution is suboptimal. In \cite{wang2020non}, authors use statistical moments of the distributions in concentration inequalities to upper-bound the chance constraint for discrete-time trajectory planning under non-Gaussian uncertainties, and solve the problem using nonlinear program solvers. A scenario-based optimization method that translates the chance constraints into deterministic ones is provided in \cite{de2021scenario}. Taking ideas from distributionally robust optimization, deterministic convex reformulation of the chance-constrained SOC problems is proposed in \cite {li2021distributionally} for discrete-time linear system dynamics. Different from the above works in this paper we consider continuous-time, control-affine system dynamics. For the continuous-time chance-constrained planning problem the authors in \cite{jasour2021convex} use risk contours to transform the original problem into a deterministic planning problem and use convex methods based on sum-of-squares optimization to obtain continuous-time trajectories. In \cite{nakka2019trajectory} a continuous-time chance-constrained SOC problem is converted to a deterministic control problem with convex constraints using generalized polynomial chaos expansion and Chebyshev inequality. The deterministic optimal control problem is then solved using sequential convex programming.  Other approaches that consider approximations of the chance constraints include approaches based on concentration of measure inequalities \cite{hokayem2013chance}, Bernstein approximation \cite{nemirovski2007convex}, and moment based surrogate \cite{paulson2019efficient}. A common limitation of the above approaches is the conservatism introduced by the approximation of the chance constraints, leading to overly cautious policies that compromise performance or cause artificial infeasibilities \cite{nemirovski2007convex, frey2020collision, patil2023upper, patil2022upper, patil2024analytical}. 
% Another way to build a computationally tractable approximation which is not necessarily conservative is based on sampling techniques \cite{blackmore2010probabilistic}, \cite{blackmore2006probabilistic}, \cite{blackmore2007robust}. The main advantage of this approach is its generality as it can be applied to arbitrary uncertainty distribution, rather than only Gaussian. However, since the approximation is based on random samples, its solution may not satisfy the chance constraints.
Deep reinforcement learning algorithms such as soft actor-critic also have been used to solve the chance-constrained SOC problems \cite{huang2021risk}.  \par
In our work, we utilize the \textit{path integral} approach to numerically solve the posed chance-constrained SOC problem using open-loop samples of system trajectories. Path integral control is a sampling-based algorithm employed to solve nonlinear SOC problems numerically ~\cite{kappen2005path, patil2025advancing, patil2023risk, patil2025path, patil2025task, patil2025advancing}. It allows the policy designer to compute the optimal control inputs online using Monte Carlo samples of system paths. The Monte Carlo simulations can be massively parallelized on GPUs, and thus the path integral approach is less susceptible to the \textit{curse of dimensionality} \cite{williams2017model}.
\subsection{Contributions}
The contributions of this work are as follows: (1) We leverage the notion of exit time from the continuous-time stochastic calculus \cite{chern2021safe, shah2011probability} to formulate a chance-constrained SOC problem. The dual of this chance-constrained SOC problem is constructed by incorporating the chance constraint into the cost function. The chance constraint is then transformed into an expectation of an indicator function, which enjoys the same additive structure as the primal cost function. No approximation nor time discretization is introduced in this transformation. (2) Given a fixed dual variable, we evaluate the dual objective function by numerically solving the Hamilton-Jacobi-Bellman (HJB) partial differential equation (PDE). (3) It is shown that under a certain assumption on the system dynamics and cost function, a strong duality holds between the primal problem and the dual problem (the duality gap is zero). (4) We propose a novel path-integral-based dual ascent algorithm to numerically solve the dual problem. This allows us to solve the original chance-constrained problem online via open-loop samples of system trajectories. Finally, we present simulation studies on chance-constrained motion planning for spatial navigation of mobile robots. The solution obtained using the path integral approach is compared with that of the finite difference method.\par

Early forms of the results in this paper were presented in \cite{patil2022chance}. This paper extends these results by (1) formulating a dual SOC problem, (2) proving that a strong duality exists between the primal problem and the dual problem under an assumption on the system dynamics and cost function, and (3) numerically solving the dual problem via path-integral-based dual ascent algorithm.  
 
\subsection*{Notation}
{\color{black}If a stochastic process ${x}(t)$ starts from $x^{\text{init}}$ at time $t$, then let $P_{x^{\text{init}},t}\left(\mathcal{E}\right)=P\big(\mathcal{E}({x})\;|\;{x}(t)=x^{\text{init}}\big)$ denote the probability of event $\mathcal{E}({x})$ involving the stochastic process ${x}(t)$ conditioned on ${x}(t)=x^{\text{init}}$, and let $\mathbb{E}_{x^{\text{init}}, t}\left[F\left({x}\right)\right]=\mathbb{E}\left[F\left({x}\right)\;|\;{x}\left(t\right)=x^{\text{init}}\right]$ denote the expectation of $F\left({x}\right)$ (a functional of ${x}\left(t\right)$) also conditioned on ${x}\left(t\right)=x^{\text{init}}$.} Let $\mathds{1}_{\mathcal{E}}$ be an indicator function, which returns $1$ when the condition $\mathcal{E}$ holds and 0 otherwise. The $\bigvee$ symbol represents a logical OR implying the existence of a satisfying event among a collection. $\text{Tr}(A)$ denotes the trace of a matrix $A$.
\section{Preliminaries}
Let $\mathcal{X}_{s}\subseteq\mathbb{R}^n$ be a bounded open set representing a safe region, $\partial\mathcal{X}_{s}$ be its boundary, and closure $\overline{\mathcal{X}_{s}}=\mathcal{X}_{s}\cup\partial\mathcal{X}_{s}$.
\subsection{Controlled and Uncontrolled Processes}\label{Sec: Controlled and Uncontrolled Processes}
Consider a controlled process ${x}(t)\in\mathbb{R}^n$ driven by following control-affine It\^{o} stochastic differential equation (SDE):
\begin{equation}\label{SDE}
\begin{aligned}
   d{x}(t)=&{f}\left({x}(t), t\right)dt+{G}\left({x}(t), t\right){u}({x}(t), t)dt\\
   &+{\Sigma}\left({x}(t), t\right)d{w}(t),
\end{aligned}
\end{equation}
where ${u}\left({x}(t), t\right)\in\mathbb{R}^m$ is a control input, ${w}(t)\in\mathbb{R}^k$ is a $k$-dimensional standard Wiener process on a suitable probability space $\left(\Omega, \mathcal{F}, P\right)$, ${f}\left({x}(t), t\right)\in\mathbb{R}^n$, ${G}\left({x}(t), t\right)\in\mathbb{R}^{n\times m}$ and ${\Sigma}\left({x}(t), t\right)\in\mathbb{R}^{n\times k}$. {\color{black} Without loss of generality, we assume that $\Sigma$ is left-invertible; that is, it has full column rank.} Let $\hat{{x}}(t)\in\mathbb{R}^n$ be an uncontrolled process driven by the following SDE:
\begin{equation}\label{uncontrolled SDE}
  d\hat{{x}}(t)\!=\!\!{f}\!\left(\hat{{x}}(t),\! t\right)\!dt\!+\!{\Sigma}\!\left(\hat{{x}}(t),\! t\right)\!d{w}(t). 
\end{equation}
At the initial time $t_0$, ${x}(t_0)=\hat{{x}}(t_0)=x_0\in\overline{\mathcal{X}_s}$. Throughout this paper, we assume sufficient regularity in the coefficients of (\ref{SDE}) and (\ref{uncontrolled SDE}) so that unique strong solutions exist \cite[Chapter 1]{oksendal2013stochastic}. In the rest of the paper, for notational compactness, the functional dependencies on $x$ and $t$ are dropped whenever it is unambiguous. 
\subsection{Probability of Failure}
For a given finite time horizon $t\in[t_0, T]$, $t_0<T$, if the system (\ref{SDE}) leaves the safe region $\mathcal{X}_{s}$ at any time $t\in(t_0, T]$, then we say that it fails. 
\begin{definition}[Probability of failure]
 The probability of failure $P_\mathrm{fail}\left(x_0,t_0,u(\cdot)\right)$ of system (\ref{SDE}) starting at $(x_0, t_0)$ and operating under the policy $u(\cdot)$ is defined as   
\begin{equation}\label{pfail}
    P_\mathrm{fail}\left(x_0,t_0,u(\cdot)\right):=P_{x_0,t_0}\left(\bigvee_{t\in(t_0, T]} {x}(t)\notin \mathcal{X}_{s}\right).
\end{equation}
\end{definition}
\subsection{Exit Times}
Let $\mathcal{Q}=\mathcal{X}_s\times[t_0, T)$ be a bounded set with the boundary  $\partial\mathcal{Q}=\left(\partial\mathcal{X}_s\times[t_0,T]\right)\cup\left(\mathcal{X}_s\times\{T\}\right)$, and closure $\overline{\mathcal{Q}}=\mathcal{Q}\cup\partial\mathcal{Q}=\overline{\mathcal{X}_s}\times[t_0, T]$. We define exit time ${t}_{f}$ for process ${x}(t)$ as
\begin{equation}\label{tf with Q}
 {t}_{f} \coloneqq \text{inf}\{t> t_0: ({x}(t), t)\notin \mathcal{Q}\}.   
\end{equation}
Alternatively, ${t}_f$ can be defined as 
\begin{equation}\label{tf}
{t}_{f} \coloneqq 
\begin{cases}
T, & \!\!\!\!\!\!\!\!\!\!\!\!\!\!\!\!\!\!\!\!\!\!\!\!\!\!\!\!\!\!\!\!\!\!\!\!\!\!\text{if}\;\; {x}(t)\in\mathcal{X}_{s}, \forall t\in(t_0, T),\\
\text{inf}\;\{t\in(t_0, T) : {x}(t)\notin\mathcal{X}_{s}\}, & \text{otherwise}.
 \end{cases}
\end{equation}
For one-dimensional state space, the domains $\mathcal{X}_s$, $\mathcal{Q}$ and their boundaries $\partial \mathcal{X}_s$, $\partial\mathcal{Q}$ are shown in Figure \ref{Fig. computational domain}. The figure also depicts the exit times for two realizations of trajectories $\{x(t), t\in[t_0, {t}_f]\}$. 
\begin{figure}
    \centering
    \!\!\!\!\!\!\!\!\!\!\!\!\includegraphics[scale=0.28]{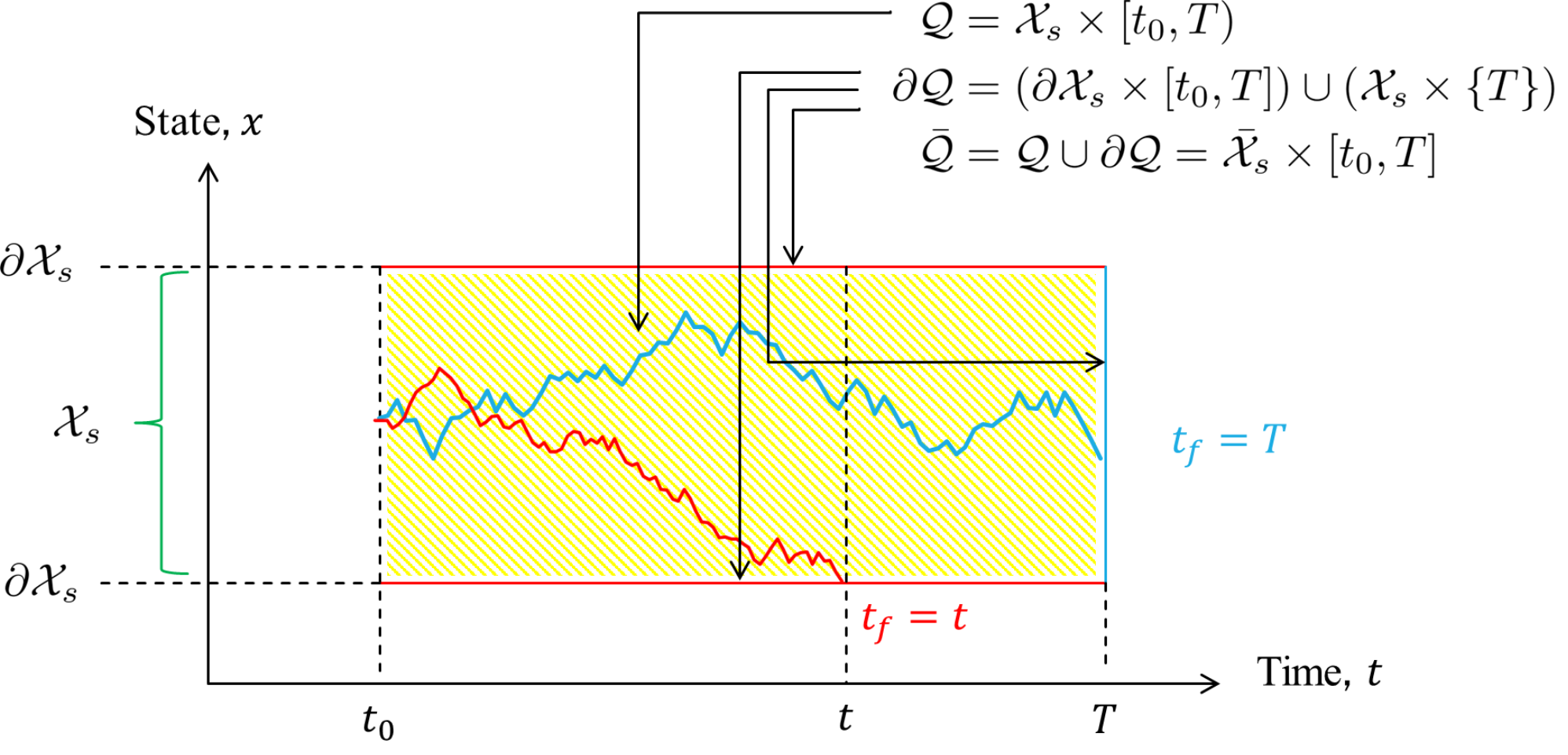} 
      
        \caption{Computational domains and exit times $t_f$} 
        \label{Fig. computational domain}
\end{figure}

Similarly, exit time $\hat{{t}}_f$ for process $\hat{{x}}(t)$ is defined as

\begin{equation}\label{t_hat_f with Q}
 \hat{{t}}_{f} \coloneqq \text{inf}\{t> t_0: (\hat{{x}}(t), t)\notin \mathcal{Q}\}.   
\end{equation}

Note that by the above definitions, $\left({x}({t}_{f}),{t}_{f}\right)\in\partial{\mathcal{Q}}$ and $\left(\hat{{x}}(\hat{{t}}_{f}),\hat{{t}}_{f}\right)\in\partial{\mathcal{Q}}$.
\section{Problem Formulation}
We first formalize a chance-constrained SOC problem in Section \ref{Sec: CC-SOC}. In Section \ref{Sec: dual}, we formulate its dual.

\subsection{Chance-Constrained SOC Problem}\label{Sec: CC-SOC}
Consider a cost function that is quadratic in the control input and has the form:
\begin{equation}\label{C}
\begin{split}
 C\left(x_0, t_0, {u}(\cdot)\right)&\coloneqq\mathbb{E}_{x_0,t_0}\Bigg[\psi\left({x}({t}_{f})\right)\cdot\mathds{1} _{{x}({t}_{f})\in \mathcal{X}_{s}}+\\
 &\!\!\!\!\!\!\!\!\!\!\!\!\!\!\!\!\!\!\!\!\!\!\!\!\!\!\!\!\!\!\!\!\!\!\!\!\!\!\!\int_{t_0}^{{t}_{f}}\left(\frac{1}{2}{u}^\top{R}\left({x}(t), t\right){u}+ V\left({x}(t), t\right)\right)dt\Bigg]
\end{split}
\end{equation}
{\color{black} where sufficiently differentiable non-negative functions $\psi\left({x}({t}_{f})\right)$ and $V\left({x}(t), t\right)$ denote the terminal cost and state-dependent running cost, respectively. $V\left({x}(t), t\right)$ is also supposed to be integrable. $R\left({x}(t), t\right)\in\mathbb{R}^{m\times m}$, a symmetric positive definite matrix (for all values of ${x}(t)$ and $t$) represents the weight with respect to the control cost.} $\mathds{1} _{{x}({t}_{f})\in \mathcal{X}_{s}}$ returns $1$ if the state of the system at the exit time ${t}_f$ is inside the safe region and $0$ otherwise i.e., the terminal cost is active only when the state of the system at the exit time ${t}_f$ is safe. Note that our cost function is defined over time horizon $[t_0, {t}_{f}]$ instead of $[t_0, T]$, because we do not consider the cost incurred after the system fails. We wish to find an optimal control policy for system (\ref{SDE}) such that $C\left(x_0, t_0, u(\cdot)\right)$ is minimal and the probability of failure (\ref{pfail}) is below a specified threshold. This problem can be formulated as a chance-constrained SOC problem as follows:
\begin{problem}[Chance-constrained SOC problem]\label{Problem: Risk-constrained SOC problem}
    \begin{equation}\label{CC-SOC}
    \begin{aligned}
    \min_{{u}(\cdot)}\;& \mathbb{E}_{x_0, t_0}\!\!\left[\psi\left({x}({t}_{f})\right)\!\cdot\!\mathds{1} _{{x}({t}_{f})\in \mathcal{X}_{s}}\!+\!\!\int_{t_0}^{{t}_{f}}\!\!\!\left(\frac{1}{2}{u}^\top\!{R}{u}+\!V\!\right)\!dt\right]\\
    \emph{s.t.} \;\; & d{x}={f}dt+{G}{u}dt+{\Sigma}d{w},\quad {x}(t_0)=x_0,\\
      &P_{x_0, t_0}\left(\bigvee_{t\in(t_0, T]} {x}(t)\notin \mathcal{X}_{s}\right)\leq\Delta,
    \end{aligned}    
    \end{equation}
    where $\Delta\in(0,1)$ represents a given risk tolerance over the horizon $[t_0, T]$, and the admissible policy $u(\cdot)$ is measurable with respect to the $\sigma$-algebra generated by ${x}(s), 0\leq s\leq t$.
\end{problem}
% \begin{remark}

% % Moreover, the above formulation based on exit time ${t}_f$ is useful in order to maintain the time-additive Bellman structure of the cost function after Lagrangian relaxation.
% \end{remark}
\subsection{Dual SOC Problem}\label{Sec: dual}
We define the Lagrangian associated with Problem \ref{Problem: Risk-constrained SOC problem} as
\begin{align}\label{eq: lagrangian}
\mathcal{L}\left(x_0, t_0, {u}(\cdot); \eta\right) &=  C\left(x_0, t_0, {u}(\cdot)\right) \nonumber\\
&\!\!\!\!\!\!\!\!\!\! + \eta P_{x_0, t_0}\left(\bigvee_{t\in(t_0, T]} {x}(t)\notin \mathcal{X}_{s}\right)-\eta\Delta
\end{align}
where $\eta\geq 0$ is the \textit{Lagrange multiplier}.
Using a standard equality in probability theory \cite{durrett2019probability}, $P_{\mathrm{fail}}$ in the chance constraint of (\ref{CC-SOC}) can be transformed into an expectation of an indicator function as
\begin{equation}\label{pfail2}
    P_{x_0, t_0}\left(\bigvee_{t\in(t_0, T]} {x}(t)\notin \mathcal{X}_{s}\right)=\mathbb{E}_{x_0, t_0}\left[\mathds{1} _{{x}({t}_{f})\in \partial\mathcal{X}_{s}}\right].
\end{equation}
Here, $\mathds{1} _{{x}({t}_{f})\in \partial\mathcal{X}_{s}}$ returns $1$ if the state of the system \eqref{SDE} at the exit time is on the boundary of the safe set $\partial \mathcal{X}_s$ (i.e. the state of the system at the exit time escapes the the safe region) and $0$, otherwise. Using \eqref{pfail2} the Lagrangian \eqref{eq: lagrangian} can be reformulated as
    \begin{equation}\label{Lagrangian2}
\!\!\!\!\mathcal{L}\!\left(x_0, t_0, {u}(\cdot); \eta\right) \!= \!C\!\left(x_0, t_0, {u}(\cdot)\right) \!+\! \eta\! \left[\mathbb{E}_{x_0, t_0}\!\!\left[\mathds{1} _{{x}({t}_{f})\in \partial\mathcal{X}_{s}}\right]\!-\!\Delta\right]
\end{equation}
 Observe that for any $\eta\geq0$ if we define $\phi:\overline{\mathcal{X}_s}\to\mathbb{R}$ as\footnote{In what follows, function $\phi(x; \eta)$ sets a boundary condition for a PDE and we often need technical assumptions on the regularity of $\phi(x; \eta)$ (e.g., continuity on $\overline{\mathcal{X}_s}$) to guarantee the existence of a solution of the PDE. When such requirements are needed, we approximate  (\ref{phi(x)}) as $\phi(x; \eta) \approx \psi(x)B(x) + \eta\left(1-B(x)\right)$, where $B(x)$ is a smooth bump function on $\mathcal{X}_s$.}: 
\begin{equation}\label{phi(x)}
    \phi\left({x}; \eta\right)\coloneqq\psi\left({x}\right)\cdot \mathds{1} _{{x}\in \mathcal{X}_{s}}+\eta\cdot\mathds{1} _{{x}\in \partial\mathcal{X}_{s}}-\eta\Delta,
\end{equation}
then, the second term in (\ref{Lagrangian2}) can be absorbed in a new terminal cost function $\phi$ as follows:
\begin{equation}\label{Chat}
    \mathcal{L}\!\left(x_0,\!t_0,\! {u}(\cdot); \eta\right)=\mathbb{E}_{x_0, t_0}\!\!\left[\phi\!\left({x}({t}_{f}); \eta\right)\!+\!\!\!\int_{t_0}^{{t}_{f}}\!\!\!\left(\!\frac{1}{2}{u}^\top\!R{u}+V\!\right)\!dt\!\right].
\end{equation}
Now we formulate the dual problem as follows:
\begin{problem}[Dual SOC problem]\label{prob: dual problem}
\begin{align}\label{eq: dual}
     \max_{\eta} \min_{u(\cdot)}&\;\;\mathcal{L}\left(x_0, t_0, {u}(\cdot);\eta\right)\\
     \emph{s.t.}&\;\; \eta\geq 0\nonumber\\
     &\;\; d{x}={f}dt+{G}{u}dt+{\Sigma}d{w},\quad {x}(t_0)=x_0,\nonumber
\end{align}
%      \min_{\eta\geq0}\big( \min_{u(\cdot)}\mathcal{L}\left(x_0, t_0, {u}(\cdot); \eta\right) + \eta\left[\mathbb{E}_{x_0, t_0}\left[\mathds{1} _{{x}({t}_{f})\in \partial\mathcal{X}_{s}}\right]-\Delta\right]\right)&=
\end{problem}
In order to solve Problem \ref{prob: dual problem} we first solve the subproblem 
\begin{equation}\label{eq: g of eta}
g(\eta)\coloneqq \min_{u(\cdot)}\mathcal{L}\left(x_0, t_0, {u}(\cdot); \eta\right).
\end{equation}
Note that, $\mathcal{L}$ possesses the time-additive Bellman structure i.e., for any $\eta\geq0$ and $t_0\leq t\leq{t}_f$, we can write 
\begin{align*}
    \mathcal{L}\left(x_0,t_0, {u}(\cdot); \eta\right)=&\mathbb{E}_{x_0, t_0}\left[\int_{t_0}^{t}\!\!\left(\frac{1}{2}{u}^\top R{u}+V\right){\color{black}ds}\right]\\
    &\!\!\!\!\!\!\!\!\!\!\!\!\!\!\!\!\!\!\!\!\!\!\!\!+ \mathbb{E}_{x, t}\left[\phi\left({x}({t}_{f}); \eta\right)+\int_{t}^{{t}_{f}}\!\!\left(\frac{1}{2}{u}^\top R{u}+V\right){\color{black}ds}\right].
\end{align*}
Therefore problem \eqref{eq: g of eta} can be solved by utilizing dynamic programming without having to introduce any conservative approximation of the failure probability $P_{\mathrm{fail}}$. After solving the subproblem \eqref{eq: g of eta} we solve the dual problem
\begin{equation}\label{eq: dual problem}
   \max_{\eta \geq0} g(\eta).
\end{equation}
 Since the dual function $g(\eta)$ is the pointwise infimum of a family of affine functions of $\eta$, it is concave even when the primal problem is not convex. Moreover, since affine functions are upper semicontinuous, $g(\eta)$ is also upper semicontinuous. {\color{black}We will use the upper semicontinuity of $g(\eta)$ to prove the existence of a dual optimal solution (See Lemma \ref{lem:existence_dual_sol} in Appendix A).}   
\section{Main Results}
In this section, we first express the dual function in terms of the solution to an HJB PDE parametrized by the dual variable $\eta$. Next, we show that the strong duality holds between the primal problem \eqref{CC-SOC} and the dual problem \eqref{eq: dual} (the duality gap is zero) under a certain assumption on the system dynamics and the cost function. Finally, we propose a novel Monte-Carlo-based dual ascent algorithm to numerically solve the dual problem \eqref{eq: dual problem}. This implies that the optimal control input for the original chance-constrained problem (primal problem) can be computed online by real-time Monte-Carlo simulations.

\subsection{Computation of the Dual Function}\label{Sec: Computation of the Dual Function}
In this section, we compute the dual function by solving problem \eqref{eq: g of eta} using dynamic programming. For each $(x,t)\in\overline{\mathcal{Q}}$, $\eta\geq0$ and an admissible policy $u(\cdot)$ over $[t, T)$, we define the cost-to-go function:
\begin{equation}\label{value function}
    \mathcal{L}\!\left(x, t, {u}(\cdot); \eta\right)\!=\!\mathbb{E}_{x, t}\!\!\left[\phi\left({x}({t}_{f});\eta\right)\!+\!\!\int_{t}^{{t}_{f}}\!\!\!\left(\frac{1}{2}{u}^\top\!R{u}+\!V\!\right)\!{\color{black}ds}\right]\!\!.
\end{equation}
Now, we state the following theorem, {\color{black}that is obtained by extending the verification theorem presented in \cite[Chapter IV]{fleming2006controlled}.}
\begin{theorem}\label{theorem: solution to risk-minimizing soc}
Suppose for a given $\eta$, there exists a function $J:\overline{\mathcal{Q}}\rightarrow \mathbb{R}$ such that 
\begin{enumerate}[(a)]
    \item $J(x,t; \eta)$ is continuously differentiable in $t$ and twice continuously differentiable in $x$ in the domain $\mathcal{Q}$;
    % \item \textcolor{red}{$J_\eta(x,t)$ is continuous on $\partial\mathcal{Q}$};
    \item $J(x,t;\eta)$ solves the following dynamic programming PDE (HJB PDE):
    \begin{equation}\label{HJB PDE}
  \!\!\!\!\!\!\!\!\!\!\!\!\begin{cases}
     \begin{aligned}
         \!\!-\partial_tJ\!=&\!-\frac{1}{2}\!\left(\partial_xJ\right)^\top\!\!GR^{-1}G^\top\!\partial_xJ\!+\!V\\
         &+\!\!f^\top\!\partial_xJ+\frac{1}{2}\text{Tr}\left(\Sigma\Sigma^\top\partial^2_xJ\right),
          \end{aligned} &  \forall(x,t)\in\mathcal{Q}, \\
    \!\!\underset{\substack{(x,t)\to(y,s) \\ (x,t)\in\mathcal{Q}}}{\lim}J(x,t;\eta)=\phi(y;\eta), & \!\!\!\forall(y,s)\in\partial\mathcal{Q}.
  \end{cases}
    \end{equation}
\end{enumerate}
Then, the following statements hold:
\begin{enumerate}[(i)]
\item $J(x,t;\eta)$ is the \textit{value function} for Problem \eqref{eq: g of eta}. That is,
\begin{equation}\label{J as value function}
    J\left(x, t; \eta\right) = \min_{{u}(\cdot)} \mathcal{L}\left(x, t, {u}(\cdot);\eta\right),\quad \forall\;(x,t)\in\overline{\mathcal{Q}}.
\end{equation}

\item The solution to Problem \eqref{eq: g of eta} is given by 
\begin{equation}\label{optimal policy}
    u^*(x,t; \eta)=-R^{-1}\left(x, t\right){G}^\top\left(x, t\right)\partial_xJ\left(x, t; \eta\right).
\end{equation}
\end{enumerate}
\end{theorem}

% P\left(\exists\;t\in[t_0, T] \text{ s.t. } {x}(t)\notin\mathcal{X}_{s}\right)

% \section{Main Results}\label{Sec: Safe control}
% This section presents the main results of the paper. In Section \ref{Section: Risk-Minimizing HJB PDE}, we show that Problem \ref{Problem: Risk-minimizing SOC problem} can be solved via an HJB PDE with appropriate Dirichlet boundary condition. In Section \ref{Sec: Risk Analysis}, we present a special case of the proposed risk-minimizing control problem, namely the risk estimation problem. Lastly, in Section \ref{Sec: Linearizable HJB PDE}, we find a solution of a class of risk-minimizing HJB equations that can be linearized. 

% an optimal risk-minimizing control policy can be synthesized by solving an HJB PDE with an appropriate boundary condition.

% \subsection{Risk-Minimizing HJB PDE}\label{Section: Risk-Minimizing HJB PDE}
\begin{proof}
Let $J$ be the function satisfying (a) and (b). By Dynkin's formula \cite[Theorem 7.4.1]{oksendal2013stochastic}, for each $(x,t)\in\overline{\mathcal{Q}}$ and $\eta$, we have
\begin{equation}\label{ito with dJ}
 \mathbb{E}_{x,t}\left[J\left({x}({t}_{f}), {t}_{f};\eta\right)\right]\!=\!J(x,t; \eta) +  \mathbb{E}_{x,t}\left[\int_{t}^{{t}_{f}}\!\!\!\!\!\!dJ\left({x}(s), s; \eta\right)\right]\!,
\end{equation}
where
\begin{equation}\label{dJ}
    \begin{aligned}
\!\!\!dJ\left({x}(s), s; \eta\right)=&(\partial_tJ)ds + (d{x})^\top\partial_xJ+\frac{1}{2}(d{x})^\top\!\!\left(\partial_x^2J\right)\!(d{x})\\
=&(\partial_tJ)ds + (f\!+Gu)^\top\!(\partial_xJ)ds\\
&+\left(\Sigma d{w}\right)^\top(\partial_xJ)+\!\frac{1}{2}\text{Tr}\left(\Sigma\Sigma^\top\!\partial_x^2J\right)ds.
    \end{aligned}
\end{equation}
Notice that the term $\int_{t}^{{t}_{f}}\left(\Sigma d{w}\right)^\top(\partial_xJ)$ is an It\^o integral. Using the property of It\^o integral \cite[Chapter 3]{oksendal2013stochastic}, we get 
\begin{equation}\label{ito integral}
    \mathbb{E}_{x,t}\int_{t}^{{t}_{f}}\left(\Sigma d{w}\right)^\top(\partial_xJ)=0
\end{equation}
Substituting (\ref{dJ}) into (\ref{ito with dJ}) and using \eqref{ito integral}, we have  
\begin{equation}\label{ito}
    \begin{aligned}
       \mathbb{E}_{x,t}\left[J\left({x}({t}_{f}), {t}_{f};\eta\right)\right]&=J(x,t;\eta)\\ &\!\!\!\!\!\!\!\!\!\!\!\!\!\!\!\!\!\!\!\!\!\!\!\!\!\!\!\!\!\!\!\!\!\!\!\!\!\!\!\!\!\!\!\!\!\!\!\!\!\!\!\!+\mathbb{E}_{x,t}\left[\int_{t}^{{t}_{f}}\!\!\!\left(\!\partial_tJ + (f+Gu)^{\!\top}\!(\partial_xJ)+\frac{1}{2}\text{Tr}\left(\Sigma\Sigma^{\!\top}\!\partial_x^2J\right)\!\right)\!ds\right]\!\!.
    \end{aligned}
\end{equation}
By the boundary condition of the PDE (\ref{HJB PDE}),\newline $J\left({x}({t}_{f}), {t}_{f}; \eta\right) =\phi\left({x}({t}_{f});\eta\right)$. Hence, from (\ref{ito}), we obtain
\begin{equation}\label{ito after plugging BC}
    \begin{aligned}
       J(x,t;\eta)&=\mathbb{E}_{x,t}\left[\phi\left({x}({t}_{f});\eta\right)\right]\\ &\!\!\!\!\!\!\!\!\!\!\!\!\!\!\!\!\!\!\!\!\!\!\!-\mathbb{E}_{x,t}\left[\int_{t}^{{t}_{f}}\!\!\!\left(\!\partial_tJ + (f+Gu)^{\!\top}\!(\partial_xJ)+\frac{1}{2}\text{Tr}\left(\Sigma\Sigma^{\!\top}\!\partial_x^2J\right)\!\right)\!ds\right]\!\!.
    \end{aligned}
\end{equation}
Now, notice that the right hand side of the PDE in (\ref{HJB PDE}) can be expressed as the minimum value of a quadratic form in $u$ as follows:
\begin{equation*}
    \begin{aligned}
         -\partial_tJ=&-\frac{1}{2}\left(\partial_xJ\right)^\top GR^{-1}G^\top\partial_xJ+V\\
         &+f^\top\partial_xJ+\frac{1}{2}\text{Tr}\left(\Sigma\Sigma^\top\partial^2_xJ\right)\\
         &\!\!\!\!\!\!\!\!\!\!\!\!\!\!\!\!\!\!\!\!\!\!=\min_{{u{\color{black}(\cdot)}}}\left[\frac{1}{2}u^\top Ru+V+\left(f+Gu\right)^\top\partial_xJ
         +\frac{1}{2}\text{Tr}\left(\Sigma\Sigma^\top\partial^2_xJ\right)\right]\!\!.
          \end{aligned} 
\end{equation*}
Therefore, for an arbitrary $u$, we have 
\begin{equation}\label{pde ineqaulity}
    \begin{aligned}
         \!\!\!\!-\partial_tJ\!\leq\!\frac{1}{2}u^\top\!Ru\!+\!V\!+\!\left(f\!+\!Gu\right)^\top\!\!\partial_xJ\!
         +\!\frac{1}{2}\text{Tr}\!\left(\Sigma\Sigma^\top\partial^2_xJ\right)
          \end{aligned} 
\end{equation}
where the equality holds iff
\begin{equation}\label{optimal policy2}
    u = -R^{-1}G^\top\partial_xJ.
\end{equation}
Combining (\ref{ito after plugging BC}) and (\ref{pde ineqaulity}), we obtain
\begin{equation}\label{J leq C_hat}
  \begin{aligned}
       \!J(x,t; \eta)\!&\leq\!\mathbb{E}_{x,t}\!\left[\phi\!\left({x}({t}_{f}); \eta\right)\right]\!+\!\mathbb{E}_{x,t}\!\left[\!\int_{t}^{{t}_{f}}\!\!\!\left(\!\frac{1}{2}{u}^\top\!R{u}\!+\!V\!\!\right)\!ds\right]\\
       &= \mathcal{L}\left(x,t,{u}(\cdot);\eta\right).
    \end{aligned}   
\end{equation}
This proves the statement (i). Since (\ref{J leq C_hat}) holds with equality iff (\ref{optimal policy2}) is satisfied, the statement (ii) also follows. 
\end{proof}
% \begin{remark}
%    General conditions under which there exists a function $J(x,t; \eta)$ satisfying (a) and (b) in Theorem \ref{theorem: solution to risk-minimizing soc} are beyond the scope of this paper. However, below we focus on a special case in which (\ref{HJB PDE}) can be linearized, where the existence and uniqueness of such function is guaranteed.  
% \end{remark}
Theorem \ref{theorem: solution to risk-minimizing soc} implies that the solution of problem \eqref{eq: g of eta} can be expressed in terms of the HJB PDE \eqref{HJB PDE} parameterized by the dual variable $\eta$. Notice that \eqref{HJB PDE} is a nonlinear PDE, which is in general, difficult to solve. In what follows, we linearize the PDE \eqref{HJB PDE} whose solution can be obtained by utilizing the Feyman-Kac lemma \cite{williams2017model}. First, we make the following assumption which is essential to linearize the PDE \eqref{HJB PDE}.
\begin{assumption}\label{Assum: linearity}
For all $(x,t)\in\overline{\mathcal{Q}}$, there exists a positive constant $\lambda$ satisfying the following equation: 
\begin{equation}\label{lambda}
  \Sigma(x, t)\Sigma^\top(x, t) = \lambda G(x, t)R^{-1}(x,t) G^\top(x, t). 
\end{equation}
\end{assumption}
The above assumption implies that the control input in the direction with higher noise variance is cheaper than that in the direction with lower noise variance. See \cite{kappen2005path} for further discussion on this condition. 

 {\color{black}\begin{remark}
The following analysis relies critically on Assumption \ref{Assum: linearity}, which ensures strong duality between the primal problem \eqref{CC-SOC} and the dual problem \eqref{eq: dual}. To satisfy this assumption—specifically, to preserve the relationship between control authority and system noise—two key conditions must be met: (i) the stochasticity has to enter the dynamics \eqref{SDE} via the control channel and (ii) the control cost has to be tuned properly to match the variance of the stochastic forces.
Condition (ii) is less restrictive because the cost function under minimization has additional state-dependent terms. These terms can be tuned to capture a large envelope of cost functions and behaviors, despite the fact that the control cost weight is determined based on the variance of the process noise. In contrast, condition (i) is more restrictive, particularly in dynamical systems where the source of stochasticity is not only due to additive stochastic forces but also due to model uncertainty in indirectly actuated states. Nonetheless, it still encompasses a broad class of SOC problems. 
 \end{remark}}

Suppose Assumption 1 holds. Using the constant $\lambda$ that satisfies (\ref{lambda}), we introduce the following transformed value function $\xi(x,t;\eta)$:
\begin{equation}\label{exp transformation}
 \xi(x,t;\eta) = \exp\left(-\frac{1}{\lambda}J(x,t;\eta)\right).
\end{equation}
Transformation (\ref{exp transformation}) allows us to write the PDE (\ref{HJB PDE}) in terms of $\xi\left(x,t;\eta\right)$ as:

\begin{equation}\label{transformed HJB PDE}
  \!\!\!\!\!\!\begin{cases}
     \begin{aligned}
         \!\partial_t\xi\!=&\frac{V\xi}{\lambda}-\!\frac{1}{2}\text{Tr}\!\left(\Sigma\Sigma^\top\!\partial^2_x\xi\right)\!+\!\frac{1}{2\xi}\!\left(\partial_x\xi\right)^{\!T}\!\!\Sigma\Sigma^\top\!\partial_x\xi\\
         &\!\!\!\!\!\!\!\!\!\!\!\!\!\!\!\!-\!\frac{\lambda}{2\xi}\!\left(\partial_x\xi\right)^\top\!\!\left(GR^{-1}G^\top\right)\!\partial_x\xi\!-\!f^\top\partial_x\xi,
          \end{aligned} &  \\
          & \!\!\!\!\!\!\!\!\!\!\!\!\!\!\!\!\!\!\!\!\!\!\!\!\!\!\!\!\!\!\forall(x,t)\in\mathcal{Q}, \\
    \!\!\underset{\substack{(x,t)\to(y,s) \\ (x,t)\in\mathcal{Q}}}{\lim}\xi(x,t;\eta)\!=\!\text{exp}{\left(-\frac{\phi(y;\eta)}{\lambda}\right)}, & \!\!\!\!\!\!\!\!\!\!\!\!\!\!\!\!\!\!\!\!\!\!\!\!\!\!\!\!\!\!\forall(y,s)\in\partial\mathcal{Q}.
  \end{cases}
    \end{equation}
Using Assumption 1 in the equation \eqref{transformed HJB PDE}, we rewrite PDE (\ref{HJB PDE}) as a linear PDE in terms of $\xi\left(x,t;\eta\right)$ as:

\begin{equation}\label{linearized risk-minimizing HJB}
 \!\!\begin{cases}
     \!\partial_t\xi\!=\!\frac{V\xi}{\lambda}\!-\!f^\top\partial_x\xi-\frac{1}{2}\text{Tr}\left(\Sigma\Sigma^\top\partial^2_x\xi\right),       & \forall(x,t)\!\in\!\mathcal{Q}, \\
    \!\!\underset{\substack{(x,t)\to(y,s) \\ (x,t)\in\mathcal{Q}}}{\lim}\xi(x,t; \eta)\!=\!\text{exp}{\left(-\frac{\phi(y; \eta)}{\lambda}\right)}, & \!\!\!\forall(y,s)\!\in\!\partial\mathcal{Q}.\\  
  \end{cases}
\end{equation}
Now we find the solution of the linearized PDE \eqref{linearized risk-minimizing HJB} using the Feynman-Kac lemma.
\begin{lemma}[Feynman-Kac lemma]\label{theorem: E and !}
    At any $\eta\geq0$, the solution to the linear PDE \eqref{linearized risk-minimizing HJB} exists. Moreover, the solution is unique in the sense that $\xi$ solving \eqref{linearized risk-minimizing HJB} is given by
    \begin{equation}\label{xi}
       \begin{aligned}
  \!\!\!\!\!\!\xi\!\left(x,t;\eta\right) & \!= \!\mathbb{E}_{x,t}\!\! \left[\text{exp}\!\left(\!\!-\frac{\phi\left(\hat{{x}}({\hat{{t}}_{f}});\eta\right)}{\lambda}\!-\!\frac{1}{\lambda}\!\int_{t}^{\hat{{t}}_{f}}\!\!\!\!V\!\!\left(\hat{{x}}(r), r\right)\!dr \!\!\right)\!\right]\!.\\
  \end{aligned}
    \end{equation}
    where $\hat{{x}}(t)$ evolves according to the dynamics \eqref{uncontrolled SDE} starting at $(x,t)$ and the exit time $\hat{{t}}_f$ is defined according to \eqref{t_hat_f with Q}.
    
    \begin{proof}
    The proof follows from \cite[Theorem 9.1.1 {\color{black}and Section 9.2}]{oksendal2013stochastic}.
    \end{proof}

\end{lemma}
Now we state the following theorem which proves that under Assumption \ref{Assum: linearity}, for a given value of $\eta$, the optimal policy of problem \eqref{eq: g of eta} exists and is unique.
\begin{theorem}\label{theorem: E and ! of policy}
    If there exists a positive constant $\lambda$ that satisfies Assumption 1, then for any $\eta\geq0$ a \textit{unique} value function $J(x,t; \eta)$ of the Problem \eqref{eq: g of eta} exists and is given by $J(x,t;\eta) = -\lambda\,\text{log}\left(\xi\left(x,t;\eta\right)\right)$ where $\xi(x,t;\eta)$ is given by \eqref{xi}. Consequently, there exists a \textit{unique} optimal policy given by \eqref{optimal policy}. 
% \begin{equation}\label{u*}
%     u^*(x,t;\eta)=-R^{-1}\left(x, t\right){G}^\top\left(x, t\right)\partial_xJ\left(x, t; \eta\right).
% \end{equation}
\end{theorem}
\begin{proof}
According to Lemma \ref{theorem: E and !}, the solution of the linear PDE \eqref{linearized risk-minimizing HJB} exists and is unique in the sense that $\xi$ solving \eqref{linearized risk-minimizing HJB} is given by \eqref{xi}. Therefore $J(x,t;\eta)$ is unique and is given by \eqref{exp transformation}. Consequently, from Theorem \ref{theorem: solution to risk-minimizing soc}, a unique optimal policy exists and is given by \eqref{optimal policy}.
\end{proof} 

\subsection{Strong Duality}
Due to the weak duality \cite{boyd2004convex}, the value of the dual problem \eqref{eq: dual problem} is always a lower bound for the primal problem \eqref{CC-SOC}. The difference between the values of \eqref{eq: dual problem} and \eqref{CC-SOC} is called the duality gap. When the duality gap is zero, we say that the strong duality holds. In our current study on chance-constrained SOC, strong duality has a practical significance as it implies that an optimal solution to the ``hard-constrained" problem \eqref{CC-SOC} can be obtained by solving the ``soft-constrained" problem \eqref{eq: g of eta}, provided that the dual variable $\eta$ is properly chosen.\par
Since \eqref{CC-SOC} is a non-convex optimization problem in general, establishing the strong duality between \eqref{CC-SOC} and \eqref{eq: dual problem} is not trivial. In Appendix A, we outline general conditions under which a non-convex optimization problem admits a zero duality gap. In the sequel, we apply the result in Appendix A to \eqref{CC-SOC} to delineate the conditions under which the chance-constrained SOC admits the strong duality.\par 
The following assumption is reminiscent of the Slater's condition, which is often a natural premise for strong duality:
\begin{assumption}\label{assum: strict feasibility}
There exists a policy $\widetilde{u}(\cdot)$ such that $P_\text{fail}(x_0, t_0, \widetilde{u}(\cdot))-\Delta<0$, i.e., Problem \ref{Problem: Risk-constrained SOC problem} is strictly feasible.
\end{assumption}
\vspace{3mm}
Using Lemma \ref{lem:existence_dual_sol} in the Appendix A, we can show that under Assumptions \ref{Assum: linearity} and \ref{assum: strict feasibility}, there exists a dual optimal solution $\eta^*$ such that $0\leq\eta^*<\infty$ such that $g(\eta^*)=\underset{{\eta\geq 0}}{\max}\;g(\eta)$.
To proceed further, we need the following assumption. We conjecture that this assumption is valid under mild conditions; a formal analysis is postponed as future work. 
\begin{assumption}\label{assum: continuity of Pfail}
    Suppose Assumption 1 holds so that for each $\eta\geq 0$ a unique optimal policy $u^*(x,t;\eta)$ of Problem \eqref{eq: g of eta} exists (c.f., Theorem \ref{theorem: E and ! of policy}). We further assume that the function  $\eta\mapsto P_{\mathrm{fail}}(x_0,t_0, u^*(x,t;\eta)):[0, \infty)\rightarrow[0,1]$ is continuous.
\end{assumption}
\vspace{3mm}
 Now notice that under Assumption  \ref{Assum: linearity}, the solution $u^*(\cdot;\eta)$ of the problem ${\argmin}_{u(\cdot)}\; \mathcal{L}\left(x_0, t_0, {u}(\cdot); \eta\right)$ exists and is unique for each $\eta\geq0$. Therefore, statements (i) and (ii) of the Assumption \ref{asmp:continuity} in the Appendix A hold true. Moreover, by the Assumption \ref{assum: continuity of Pfail}, statement (iii) of Assumption \ref{asmp:continuity} in the Appendix A is satisfied. Now under Assumption \ref{asmp:continuity} using Lemma \ref{lem:comp_slackness}, we can prove the following complementary slackness statements:
 \begin{enumerate}[(a)]
        \item If $\eta^*=0$, then $P_{\mathrm{fail}}(x_0,t_0, u^*(\cdot\;;{\eta^*}){\color{black})}-\Delta\leq0$
        \item If $\eta^*>0$, then $P_{\mathrm{fail}}(x_0,t_0, u^*(\cdot\;;{\eta^*}))-\Delta=0$
    \end{enumerate}
The following theorem is the main result of this section.
\begin{theorem}\label{theorem: strong duality}
Consider problem \ref{Problem: Risk-constrained SOC problem} and suppose Assumptions \ref{Assum: linearity}, \ref{assum: strict feasibility} and \ref{assum: continuity of Pfail} hold. Then there exists a dual optimal solution $0\leq\eta^*<\infty$ that maximizes $g(\eta)$ and a unique optimal policy $u^*(\cdot\;;\eta^*)$ of problem \eqref{eq: g of eta} is an optimal policy of Problem \ref{Problem: Risk-constrained SOC problem} (primal problem) such that $C\left(x_0, t_0, {u}^*(\cdot\;;\eta^*)\right)=g(\eta^*)$ i.e., the duality gap is zero.
\end{theorem}
\begin{proof}
%     By Lemma \ref{Lemma: complementary slackness}, any optimal policy of problem \eqref{eq: g of eta} $u^*(\cdot\;;\eta^*)$ at a dual optimal solution $\eta=\eta^*$ satisfies $P_{\mathrm{fail}}(x_0,t_0, u^*(\cdot\;;\eta^*))-\Delta\leq0$. Therefore, $u^*(\cdot\;;\eta^*)$ is a primal feasible policy. Also, it follows from Lemma \ref{Lemma: complementary slackness} that $\eta^*(P_{\mathrm{fail}}(x_0,t_0, u^*(\cdot\;;\eta^*))-\Delta)=0$, i.e., the complementary slackness holds. Therefore, 
%     \begin{equation}
%         \begin{aligned}
%             g(\eta^*)= & {C}\!\left(x_0,\!t_0,\! {u}^*(\cdot\;;\eta^*)\right)+ \eta^* (P_{\mathrm{fail}}(x_0, t_0, {u}^*(\cdot\;;\eta^*))-\Delta)\\
%             =& {C}\!\left(x_0,\!t_0,\! {u}^*(\cdot\;;\eta^*)\right)
%         \end{aligned}
%     \end{equation}
% The first equality follows from the definition of $g(\eta)$. The last equality is true because the complementary slackness holds. This proves that the strong duality holds and $u^*(\cdot\;;\eta^*)$ is a primal optimal solution. Therefore, an optimal policy of problem \eqref{eq: g of eta} at a dual optimal solution $\eta=\eta^*$ is an optimal policy of the chance-constrained problem \eqref{CC-SOC}.
Refer to proof of Theorem \ref{theorem: strong duality2} in the Appendix A.
\end{proof}
 % Similar to our approach, the work presented in \cite{ono2015chance} solves the chance-constrained problem by formulating its dual. However, in this work, the joint chance constraint is conservatively approximated using Boole's inequality. Hence, the duality gap is nonzero and the obtained solution is suboptimal. Unlike this method, in our work, we prove that the strong duality exists between the primal chance-constrained problem \eqref{CC-SOC} and its dual \eqref{eq: dual} under certain assumptions on the system dynamics and cost function. Consequently, the chance-constrained problem can be solved by evaluating the dual objective function (c.f. Section \ref{Sec: Computation of the Dual Function}) and solving the dual problem (c.f. Section \ref{Sec: Solution of the Dual Problem}). 

 % {\color{black}\begin{remark}
 %    The fundamental problem settings in \cite{ono2015chance} and in our paper differ in several key aspects. First, \cite{ono2015chance} addresses discrete-time problems, whereas our formulation is developed for continuous-time systems. Second, \cite{ono2015chance} approximates the joint chance constraint conservatively using Boole’s inequality, while we do not introduce such a relaxation. Third, we consider a variable-end-time optimal control problem, whereas \cite{ono2015chance} focuses on fixed-horizon formulations. Due to these structural differences, a direct methodological comparison between \cite{ono2015chance} and our approach is not straightforward.
 % \end{remark}}
 
 In order to solve the dual problem, it is natural to use the gradient ascent algorithm $\eta \leftarrow \eta+\gamma(P_{\mathrm{fail}}(x_0,t_0,u^*(\cdot;\eta))-\Delta)$ to iteratively update the dual variable $\eta$. Here, $\gamma$ is the step size, $u^*(.;\eta)$ is the optimal policy solving \eqref{eq: g of eta}, and $P_{\mathrm{fail}}(x_0,t_0,u^*(\cdot;\eta))$ is the probability of failure under the policy $u^*(.;\eta)$. For the dual ascent algorithm to be practical, we need to be able to evaluate the gradient $P_{\mathrm{fail}}-\Delta$ efficiently in each iteration. Therefore, next, we study how to compute $P_{\mathrm{fail}}$ for the given value of $\eta$. 
 
 % we use the gradient ascent algorithm. For a given value of $\eta$, at each step of the gradient ascent algorithm, we need to solve the PDE \eqref{linearized risk-minimizing HJB} and find $u^*(\cdot;\eta)$. Then we need to find the probability of failure $P_\mathrm{fail}\left(x_0,t_0,u^*(\cdot;\eta)\right)$ to check if the chance constraint is satisfied or not. Therefore, next, we explain how to find the risk (failure probability) associated with the given optimal policy $u^*(\cdot)$.

\subsection{Risk Estimation}\label{Sec: Risk Analysis}
% In this Section, we compute the probability of failure for a given optimal control policy $u^*(\cdot)$.
% \begin{problem}[Risk estimation]\label{Problem: Risk estimation}
% Find the probability of failure $P_\mathrm{fail}\left(x_0,t_0,u^*(\cdot)\right)$ for system (\ref{SDE}), given the optimal control policy ${u}^*\left(\cdot\right)$.
% \end{problem}
%  We present two approaches namely, the PDE-based approach and the sampling-based approach to solve the Problem \ref{Problem: Risk estimation}. 
We present two approaches to compute $P_{\mathrm{fail}}(x_0,t_0,u^*(.;\eta))$ for a given value of $\eta$. The first approach is PDE-based in which we find the optimal policy $u^*(.;\eta)$ first, and then evaluate its risk $P_{\mathrm{fail}}(x_0,t_0,u^*(.;\eta))$ by solving a PDE. Despite conceptual simplicity, this approach is difficult to implement computationally unless there exists an analytical expression of $u^*(.;\eta)$. To circumvent this difficulty, we also present an importance-sampling-based approach. This approach allows us to numerically compute  $P_{\mathrm{fail}}(x_0,t_0,u^*(.;\eta))$ without ever constructing $u^*(.;\eta)$ and is more computationally amenable. 
 \subsubsection{PDE-Based Approach}
 We state the following theorem:
\begin{theorem}\label{Theorem: risk estimation}
Suppose there exists a function $J:\overline{\mathcal{Q}}\rightarrow \mathbb{R}$ such that 
\begin{enumerate}[(a)]
    \item $J(x,t)$ is continuously differentiable in $t$ and twice continuously differentiable in $x$ in the domain $\mathcal{Q}$;
    % \item \textcolor{red}{$J(x,t)$ is continuous on $\partial\mathcal{Q}$};
    \item $J(x,t)$ solves the following PDE for the given optimal control policy $u^*(\cdot)$:
    \begin{equation}\label{risk PDE}
      \!\!\!\!\!\!\!\!\!\!\!\!\!\begin{cases}
 \begin{aligned}
    \!\!-\partial_{t}J\!=\!\left(\!f\!+\!Gu^*\right)^\top\!\!\partial_xJ\!+\!\frac{1}{2}\text{Tr}\!\left(\!\Sigma\Sigma^\top\!\partial^2_xJ\right)\!,
    \end{aligned}& \!\!\forall(x,t)\!\in\!\mathcal{Q}, \\
    \!\!\underset{\substack{(x,t)\to(y,s) \\ (x,t)\in\mathcal{Q}}}{\lim}J(x,t)= \phi'(y), & \!\!\!\!\!\forall(y,s)\!\in\!\partial\mathcal{Q}.
    \end{cases}
          \end{equation}
\end{enumerate}
where ${\phi}'(x) \!=\! \mathds{1} _{{x}\in \partial\mathcal{X}_{s}}$. Then, $P_{\mathrm{fail}}$ is given by
\begin{equation*}
    P_\mathrm{fail}\left(x_0,t_0,u^*(\cdot)\right) = J(x_0,t_0).
\end{equation*}
\end{theorem}

\begin{proof}
    Let $J(x,t)$ be the function satisfying (a) and (b). By Dynkin's formula, for each $(x,t)\in\overline{\mathcal{Q}}$ we have
\begin{equation*}
    \begin{aligned}
       \mathbb{E}_{x,t}\left[J\left({x}({t}_{f}), {t}_{f}\right)\right]&=J(x,t)\\ &\!\!\!\!\!\!\!\!\!\!\!\!\!\!\!\!\!\!\!\!\!\!\!\!\!\!\!\!\!\!\!\!\!\!\!\!\!\!\!\!\!\!\!\!\!\!\!\!\!\!\!+\!\mathbb{E}_{x,t}\!\!\left[\!\int_{t}^{{t}_{f}}\!\!\!\left(\!\!\partial_tJ \!+\! (f\!+\!Gu^*)^\top\!(\partial_xJ)\!+\!\frac{1}{2}\text{Tr}\!\left(\!\Sigma\Sigma^\top\!\partial_x^2J\right)\!\!\right)\!ds\!\right]\!.
    \end{aligned}
\end{equation*}
The second term on the right side contains, in parentheses, the PDE in (\ref{risk PDE}) and is therefore zero. Hence, we obtain
\begin{equation}\label{flag1}
       J(x,t)=\mathbb{E}_{x,t}\left[J\left({x}({t}_{f}), {t}_{f}\right)\right].
\end{equation}
From the boundary condition of the PDE (\ref{risk PDE})
\begin{equation}\label{flag2}
     \!\!\mathbb{E}_{x,t}\!\left[J\!\left({x}({t}_{f}), {t}_{f}\right)\right] \!=\! \mathbb{E}_{x,t}\left[\phi'\left({x}({t}_{f})\right)\right] \!=\mathbb{E}_{x, t}\left[\mathds{1} _{{x}({t}_{f})\in \partial\mathcal{X}_{s}}\right].
 \end{equation}
Combining (\ref{flag1}), (\ref{flag2}), and  we obtain
\begin{equation*}
   J(x_0,t_0) = P_\mathrm{fail}\left(x_0,t_0,u^*(\cdot)\right). 
\end{equation*}
\end{proof}
\begin{remark}
 We do not prove here the existence of the solution of the PDE \eqref{risk PDE}, rather we suppose that a solution exists. Note that PDE (\ref{risk PDE}) is a special case of the Cauchy-Dirichlet problem. For proof of the existence of a solution to the Cauchy-Dirichlet problem, we refer the readers to \cite[Chapter 6]{friedman1975stochastic}.
\end{remark} 
Theorem \ref{Theorem: risk estimation} implies that if we have the solution for the optimal policy $u^*(\cdot)$, $P_\mathrm{fail}\left(x_0,t_0,u^*(\cdot)\right)$ can be computed by solving the PDE \eqref{risk PDE}. Next, we present an importance-sampling-based approach to find $P_{\mathrm{fail}}$ without constructing $u^*(\cdot)$.
\subsubsection{Importance-Sampling-Based Approach}\label{Sec: Importance-Sampling-Based Approach}
Let $\mathcal{T}$ be the space of trajectories $x\coloneqq\{x(t), t\in[t_0, {t}_f]\}$. Let $Q^*(x)$ be the probability distribution of the trajectories defined by system \eqref{SDE} under the optimal policy $u^*(\cdot)$. Using \eqref{pfail2}, we can write
\begin{align}\label{pfail under Q*}
P_{\mathrm{fail}}(x_0, t_0, u^*(\cdot)) &= \int_{\mathcal{T}}\mathds{1} _{{x}({t}_{f})\in \partial\mathcal{X}_{s}}Q^*(dx).
\end{align}
Suppose we generate an ensemble of large number of $N$ trajectories $ \{x^{(i)}\}_{i=1}^N$ under the distribution $Q^*$. Then according to the strong law of large numbers, as $N\rightarrow\infty$, 
\begin{align*}
\frac{1}{N} \sum_{i=1}^{N} \mathds{1} _{{x}^{(i)}({t}_{f})\in \partial\mathcal{X}_{s}} \overset{a.s.}{\rightarrow} P_{\mathrm{fail}}(x_0, t_0, u^*(\cdot)) \quad x^{(i)}\sim Q^*(x). 
\end{align*}
This implies that a Monte Carlo algorithm is applicable to numerically evaluate $P_{\mathrm{fail}}$. However, such a Monte Carlo algorithm is impractical since it is difficult to sample trajectories from $Q^*$ as the optimal policy $u^*$ is unknown. Fortunately, an importance sampling scheme is available which allows us to numerically evaluate $P_{\mathrm{fail}}$ using a trajectory ensemble sampled from the distribution $P(x)$ of the uncontrolled system (2). The next theorem provides the details:  
\begin{theorem}\label{Thm: PI fast}
Suppose we generate an ensemble of a large number of  $N$ trajectories $ \{x^{(i)}\}_{i=1}^N$ under the distribution $P$. For each $i$, let $r^{(i)}$ be the path reward of the sample path $i$ given by 
\begin{align}\label{r(i)}
    r^{(i)} &= \text{exp}\!\left(\!\!-\frac{\phi\left({{x}^{(i)}}({{{t}}_{f}});\eta\right)}{\lambda}\!-\!\frac{1}{\lambda}\!\int_{t_0}^{{{t}}_{f}}\!\!\!\!V\!\!\left({{x}^{(i)}}(r), r\right)\!dr \!\!\right)
    % x^{(i)}& \sim P(x).
\end{align}
Let us define $r \coloneqq \sum_{i=1}^{N}r^{(i)}$. Then as $N\rightarrow\infty$, 
\begin{equation*}
    \sum_{i=1}^{N}\frac{r^{(i)}}{r}\mathds{1}_{{x}^{(i)}({t}_{f})\in \partial\mathcal{X}_{s}} \overset{a.s.}{\rightarrow} P_{\mathrm{fail}}(x_0, t_0, u^*(\cdot)) \quad x^{(i)}\sim P(x).
\end{equation*}
\end{theorem}
\begin{proof}
    Notice that $P_{\mathrm{fail}}$ in \eqref{pfail under Q*} can be equivalently written as
\begin{equation}\label{pfail under P}
    P_{\mathrm{fail}}(x_0, t_0, u^*(\cdot))  = \int_{\mathcal{T}}\mathds{1} _{{x}({t}_{f})\in \partial\mathcal{X}_{s}}\frac{dQ^*}{dP}(x)P(dx)
\end{equation}
where the Radon-Nikodym derivative $\frac{dQ^*}{dP}(x)$ represents the likelihood ratio of observing a sample path $x$ under distributions $Q^*$ and $P$.  Now, according to Theorem \ref{Thm: likelihood ratio} in the Appendix C, for a given ensemble of $N$ trajectories $\{x^{(i)}\}_{i=1}^N$ sampled under the distribution $P$, the likelihood ratio $\frac{dQ^*}{dP}$ of observing a sample path $x^{(i)}$ is given by $\frac{r^{(i)}/r}{1/N}$. Therefore, using \eqref{pfail under P}, by the strong law of large numbers, as $N\rightarrow\infty$ we get
\begin{align*}
    \frac{1}{N} &\sum_{i=1}^{N}\frac{r^{(i)}/r}{1/N}\mathds{1} _{{x}^{(i)}({t}_{f})\in \partial\mathcal{X}_{s}} \overset{a.s.}{\rightarrow} P_{\mathrm{fail}}(x_0, t_0, u^*(\cdot))\\
    {x}^{(i)} &\sim P(x)
\end{align*}
which completes the proof.
% \begin{equation*}
%     \text{i.e.}\;\sum_{i=1}^{N}\frac{r^{(i)}}{r}\mathds{1}_{{x}^{(i)}({t}_{f})\in \partial\mathcal{X}_{s}} \overset{a.s.}{\rightarrow} P_{\mathrm{fail}}(x_0, t_0, u^*(\cdot)) \quad  {x}^{(i)} \sim P(x)
% \end{equation*}
\end{proof} 
Theorem \ref{Thm: PI fast} provides us a sampling-based approach to numerically compute the probability of failure $P_{\mathrm{fail}}(x_0, t_0, u^*(\cdot))$. According to Theorem \ref{Thm: PI fast}, if we sample $N$ trajectories under distribution $P$, then we can approximate the probability of failure as 
\begin{equation}\label{approx pfail under P}
   \!\!\!\!\!\! P_{\mathrm{fail}}(x_0, t_0, u^*(\cdot))\! \approx\! \sum_{i=1}^{N}\frac{r^{(i)}}{r}\mathds{1}_{{x}^{(i)}({t}_{f})\in \partial\mathcal{X}_{s}} \quad  {x}^{(i)} \sim P(x).
\end{equation}
where $r^{(i)}$ is defined by \eqref{r(i)} and $r = \sum_{i=1}^{N} r^{(i)}$. Note that sampling under distribution $P$ is easy since we only need to simulate the uncontrolled system dynamics \eqref{uncontrolled SDE}.

% Then, lines \ref{s_1} to \ref{e_1} are aimed at identifying the case where no feasible solution exists. For that purpose the algorithm solves $\Delta_{\mathrm{min}}=\min_{u(\cdot)}\mathbb{E}_{x_0, t_0}\left[\mathds{1} _{{x}({t}_{f})\in \partial\mathcal{X}_{s}}\right]$ by solving PDE \eqref{HJB PDE}. Note that the objective function of this optimization problem is the same as the failure probability. Hence, $\Delta_{\mathrm{min}}$ represents the minimum probability of failure that can be achieved by any possible policy. If $\Delta_{\mathrm{min}}$ is larger than the specified risk bound $\Delta$ then Problem \eqref{CC-SOC} is infeasible. 

\subsection{Solution of the Dual Problem}\label{Sec: Solution of the Dual Problem}
We present two numerical approaches to solve the dual problem \eqref{eq: dual problem}. The first approach is the finite difference method (FDM) which is one of the most popular approaches to numerically solve the partial differential equations. Despite its popularity, this approach suffers from the \textit{curse of dimensionality} and the computational cost grows exponentially as the state-space dimension increases. To circumvent this difficulty, we also present the path integral approach. This approach is more computationally amenable. It allows us to numerically solve the dual problem \eqref{eq: dual problem} online via open-loop samples of system trajectories. 
\subsubsection{Finite Difference Method}\label{Sec: FDM}
When the geometry of the computational domain is simple, it is straightforward to form discrete approximations to spatial differential operators with a high order of accuracy via Taylor series expansions \cite{Grossmann2007}. A finite number of grid points are placed at the interior and on the boundary of the computational domain, and the solution to the PDE is sought at these finite set of locations. Once the grid is determined, finite difference operators are derived to approximate spatial derivatives in the PDE. In this work, centered formulas are used to approximate spatial operators with up to eight-order accuracy at the interior grid points. For grid points near the boundary, asymmetric formulas that maintain the order of accuracy are used in conjunction with Dirichlet boundary condition. The use of finite difference operators yields a system of ordinary differential equations (ODEs) that is integrated in time with the desired method. The MATLAB suite of ODE solvers \cite{Shampine_1997_matlab} provides a number of efficient ODE integrators with high-temporal order, error control, and variable time-stepping that advance the solution in time.\par 

Note that FDM numerically solves the HJB PDE (\ref{linearized risk-minimizing HJB}) backward in time. The solution needs to be computed offline and the optimal control policy needs to be stored in a look-up table. For a given state and time, the optimal input for real-time control is obtained from the stored look-up table. Similar to the PDE \eqref{linearized risk-minimizing HJB}, FDM can be utilized to compute the probability of failure by numerically solving the PDE \eqref{risk PDE} backward in time. It is well known that the computational complexity and memory requirements of FDM increase exponentially as the state-space dimension increases. Therefore, FDM is intractable for systems with more than a few state variables. Moreover, this method is inconvenient for real-time implementation since the HJB PDE needs to be solved backward in time. Also, FDM computes a global solution over the entire domain $\mathcal{Q}$ even if the majority of the state-time pairs $(x,t)$ will never be visited by the actual system. 
To overcome these difficulties, we present the path integral approach which numerically solves the dual problem \eqref{eq: dual problem} in real-time.  
\subsubsection{Path Integral Approach} 

{\color{black}
In this section, we present the path integral approach for solving the dual problem \eqref{eq: dual problem}. Let \( \Sigma^{\dagger} \) denote the left pseudo-inverse of \( \Sigma \), i.e.,
\begin{equation}\label{sigma_dagger sigma}
  \Sigma^{\dagger} \Sigma = I,
\end{equation}
where \( I \) is the identity matrix of appropriate dimensions. As noted in Section~\ref{Sec: Controlled and Uncontrolled Processes}, we assume without loss of generality that \( \Sigma \) is left-invertible (i.e., full column rank). Let \( \Sigma^{\perp} \) denote a matrix whose rows span the left null space of \( \Sigma \), satisfying
\begin{equation}\label{sigma_perp}
  \Sigma^{\perp} \Sigma = \mathbf{0},
\end{equation}
where \( \mathbf{0} \) is the appropriately dimensioned zero matrix. Multiplying both sides of Assumption~\eqref{lambda} from the left by \( \Sigma^{\perp} \) and from the right by \( (\Sigma^{\perp})^\top \), we obtain:
\[
\Sigma^{\perp} \Sigma \Sigma^\top (\Sigma^{\perp})^\top = \lambda \Sigma^{\perp} G R^{-1} G^\top (\Sigma^{\perp})^\top.
\]
Since \( \Sigma^{\perp} \Sigma = \mathbf{0} \) by \eqref{sigma_perp}, and \( \lambda > 0 \) with \( R \succ 0 \), it follows that
\begin{equation}\label{sigma_perp G}
  \Sigma^{\perp} G = \mathbf{0}.
\end{equation}
Now, multiply both sides of the SDE~\eqref{SDE} from the left by \( \begin{bmatrix} \Sigma^{\dagger} \\ \Sigma^{\perp} \end{bmatrix} \) to obtain:
\begin{equation}\label{multiply Sigma hash}
\begin{bmatrix}
  \Sigma^{\dagger} \\
  \Sigma^{\perp}
\end{bmatrix} dx =
\begin{bmatrix}
  \Sigma^{\dagger} f \\
  \Sigma^{\perp} f
\end{bmatrix} dt +
\begin{bmatrix}
  \Sigma^{\dagger} G \\
  \Sigma^{\perp} G
\end{bmatrix} u\,dt +
\begin{bmatrix}
  \Sigma^{\dagger} \Sigma \\
  \Sigma^{\perp} \Sigma
\end{bmatrix} dw.
\end{equation}
Substituting \eqref{sigma_dagger sigma}, \eqref{sigma_perp}, \eqref{sigma_perp G} into \eqref{multiply Sigma hash}, we obtain the partitioned dynamics:
\begin{equation}\label{SDE partition}
\begin{bmatrix}
  \Sigma^{\dagger} \\
  \Sigma^{\perp}
\end{bmatrix} dx =
\begin{bmatrix}
  \Sigma^{\dagger} f \\
  \Sigma^{\perp} f
\end{bmatrix} dt +
\begin{bmatrix}
  \Sigma^{\dagger} G \\
  \mathbf{0}
\end{bmatrix} u\,dt +
\begin{bmatrix}
  I \\
  \mathbf{0}
\end{bmatrix} dw.
\end{equation}
Thus, under Assumption~\ref{Assum: linearity}, the system~\eqref{SDE} is decomposed into components that are directly and indirectly affected by the noise, as shown in~\eqref{SDE partition}.
}

The path integral control framework leverages the fact that the solution to the PDE~\eqref{linearized risk-minimizing HJB} admits a Feynman-Kac representation, as given in~\eqref{xi}. The optimal control input \( u^*(x,t;\eta) \), defined in~\eqref{optimal policy}, can then be obtained by differentiating~\eqref{xi} with respect to the state variable \( x \) \cite{williams2017model, theodorou2010generalized}. This yields:
{\color{black}\begin{equation}\label{path integral control}
u^*(x,t;\eta) \!=\! \lim_{s \searrow t}\frac{\mathcal{G}(x,t)\;\mathbb{E}_{x,t} \left[ \exp\left(-\frac{1}{\lambda} S \right) \int_t^s dw(\tau) \right]}{(s - t)\;\mathbb{E}_{x,t} \left[ \exp\left(-\frac{1}{\lambda} S \right) \right]},
\end{equation}
where the matrix \( \mathcal{G}(x,t) \) is defined as
\[
\mathcal{G}(x,t) := R^{-1} \left(\Sigma^{\dagger} G\right)^\top \left(\Sigma^{\dagger} G R^{-1} (\Sigma^{\dagger} G)^\top \right)^{-1},
\]}
and \( S \) denotes the trajectory cost-to-go for the system \( \hat{x}(t) \) starting from \( (x,t) \), given by:
\[
S = \phi\left( \hat{x}(\hat{t}_f); \eta \right) + \int_t^{\hat{t}_f} V(\hat{x}(t), t)\,dt.
\]
% {\color{black} In equation \eqref{path integral control}, the left-hand side contains a ``$dt$" term. If we take it to the right-hand side, we would get ``$\frac{dw(t)}{dt}$", which is the time derivative of the Wiener process—interpreted as \textit{white noise}. Formally, we can write
% \begin{equation*}
%    \!\!\int_{0}^{t}\!\!\! u^*\!(x,t;\eta)dt\!=\!\mathcal{G}\left(x,t\right)\frac{\mathbb{E}_{x,t}\!\!\left[\!\int_{0}^{t}\text{exp}{\left(-\frac{1}{\lambda}S\right)}\Sigma^{(2)}\!\left(x,t\right)d{w}(t)\right]}{\mathbb{E}_{x,t}\left[\text{exp}{\left(-\frac{1}{\lambda}S\right)}\right]}
% \end{equation*}
% where the term $\int_{0}^{t}\hdots dw(t)$ appearing in the numerator on the right-hand side represents a stochastic It\^o integral.} \par
To evaluate expectations in (\ref{xi}) and (\ref{path integral control}) numerically, we can discretize the uncontrolled dynamics (\ref{uncontrolled SDE}) and use Monte Carlo sampling \cite{williams2017model}. Unlike FDM, the path integral framework solves PDE (\ref{linearized risk-minimizing HJB}) in the forward direction. It evaluates a solution locally without requiring knowledge of the solution nearby so that there is no need for a (global) discretization of the computational domain. For real-time control, Monte Carlo simulations can be performed in real-time in order to evaluate (\ref{path integral control}) for the current $(x,t)$. Similar to (\ref{xi}) and (\ref{path integral control}), the failure probability \eqref{approx pfail under P} can be numerically computed by Monte Carlo sampling using the importance sampling method (cf. Section \ref{Sec: Importance-Sampling-Based Approach}). The Monte Carlo simulations can be parallelized by using the Graphic Processing Units (GPUs) and thus the path integral approach is less susceptible
to the curse of dimensionality.\par
% Similar to (\ref{xi}) and (\ref{path integral control}), the expectation in \eqref{pfail2} can be evaluated by the Monte Carlo  sampling in order to compute the failure probability $P_{\mathrm{fail}}$.
Now, we present the path-integral-based dual ascent algorithm to numerically solve the dual problem \eqref{eq: dual problem}. The algorithm for that is given in Algorithm 1. The algorithm starts by dealing with a special case $\eta=0$. We compute $P_\mathrm{fail}\left(x_0,t_0,u^*(\cdot;\eta=0)\right)$ by importance-sampling approach  presented in Section \ref{Sec: Importance-Sampling-Based Approach}. If $P_\mathrm{fail}\left(x_0,t_0,u^*(\cdot;\eta=0)\right)\leq \Delta$, we find the optimal policy $u^*(\cdot;\eta=0)$ by \eqref{path integral control} and return it. Otherwise, the algorithm chooses some initial $\eta$ and iteratively updates the dual variable $\eta \leftarrow \eta+\gamma(P_{\mathrm{fail}}(x_0,t_0,u^*(\cdot;\eta))-\Delta)$ with a learning rate of $\gamma$. Once $|P_\mathrm{fail}\left(x_0,t_0,u^*(\cdot;\eta)\right)-\Delta|$ is less than the error tolerance $\epsilon$, we return the policy $u^*(\cdot;\eta)$. Thus, using Algorithm \ref{alg:cap}, we can numerically solve the original chance-constrained problem \eqref{CC-SOC} online via real-time Monte-Carlo simulations.

\begin{algorithm}
\caption{Dual ascent via path integral approach}\label{alg:cap}
\begin{algorithmic}[1]
\Require Error tolerance $\epsilon>0$, learning rate $\gamma>0$, risk tolerance $\Delta\in(0,1)$
   % $P^{\mathrm{prev}}_{\mathrm{fail}}=0$, $\delta>0$
\State Set $\eta=0$ 

\State Compute the failure probability $P_\mathrm{fail}\left(x_0,t_0,u^*(\cdot;\eta=0)\right)$ using \eqref{approx pfail under P}. \label{sol of risk pde}
\If {$P_\mathrm{fail}\left(x_0,t_0,u^*(\cdot;\eta=0)\right)\leq\Delta$}
\State Find $u^*(\cdot;\eta=0)$ using \eqref{path integral control}\label{sol of linear HJB pde}
\State Return $u^*(\cdot;\eta=0)$ 
\EndIf
% \State Solve $\Delta_{\mathrm{min}}=\min_{u(\cdot)}\mathbb{E}_{x_0, t_0}\left[\mathds{1} _{{x}({t}_{f})\in \partial\mathcal{X}_{s}}\right]$ by solving PDE \eqref{linearized risk-minimizing HJB} \label{s_1}
% \If{$\Delta_{\mathrm{min}}>\Delta$} 
% \State return Infeasible
% \EndIf \label{e_1}

\State Choose initial $\eta>0$.
{\color{black}\While{$|P_\mathrm{fail}\left(x_0,t_0,u^*(\cdot;\eta)\right)-\Delta|\geq\epsilon$} 
\State Compute the failure probability $P_\mathrm{fail}\left(x_0,t_0,u^*(\cdot;\eta)\right)$ using \eqref{approx pfail under P} \label{sol of risk pde 2}
\State $\eta \leftarrow \eta + 
\gamma(P_\mathrm{fail}\left(x_0,t_0,u^*(\cdot;\eta)\right)-\Delta)$
\EndWhile 
\State Find $u^*(\cdot;\eta)$ using \eqref{path integral control} \label{sol of linear HJB pde 2}
\State Return $u^*(\cdot;\eta)$}
\end{algorithmic}
\end{algorithm}

\section{Simulation Results}\label{Sec: Simulation}

In this section, we demonstrate the effectiveness of the proposed control synthesis framework in addressing the chance-constrained problem defined in \eqref{CC-SOC}. In Section \ref{Sec: sim_2D}, we employ a 2D state-space velocity input model to solve a mobile robot navigation problem using Algorithm \ref{alg:cap}. The solution derived from the path-integral method is compared with that of FDM. In Section \ref{Sec: sim_5D}, we extend the analysis to a car model in 5D state-space. In this case, we solve the corresponding chance-constrained problems using Algorithm \ref{alg:cap}. Due to the high computational cost of applying FDM to models with a 5D state space, we rely exclusively on the path-integral method for this example.\par

{\color{black} We emphasize that all simulation examples considered in this section assume stable uncontrolled system dynamics. Extending the proposed path integral and HJB–Feynman-Kac–based framework to systems with inherently unstable uncontrolled dynamics would require additional care, particularly in the design of the control cost function $R$, to ensure stability of both the PDE formulation and the resulting numerical solution scheme \cite{bakshi2020stabilizing, bakshi2020schrodinger}.}

\subsection{Input Velocity Model}\label{Sec: sim_2D}
The problem is illustrated in Figure \ref{Fig. sample trajs 2D} where a particle robot wants to navigate in a 2D space from a given start position (shown by a yellow star) to the origin (shown by a magenta star), by avoiding the collisions with the red obstacle and the outer boundary.\par
Let the states of the system be ${p}_x$ and ${p}_y$, the positions along $x$ and $y$, respectively. The system dynamics are given by the following SDEs:
\begin{equation}\label{velocity input model}
\begin{split}
    d{p}_x={v}_xdt+\sigma d{w}_x,\quad d{p}_y={v}_ydt+\sigma d{w}_y,
\end{split}
\end{equation}
where ${v}_x$ and ${v}_y$ are velocities along $x$ and $y$ directions, respectively. Assume 
\begin{equation*}
 v_x = \overline{v}_x + \widetilde{v}_x, \qquad  v_y = \overline{v}_y + \widetilde{v}_y, 
\end{equation*}
where $\overline{v}_x$ and $\overline{v}_y$ are nominal velocities given by $\overline{v}_x = -k_xp_x $ and $\overline{v}_y = -k_yp_y$ for some constants $k_x$ and $k_y$. Hence, (\ref{velocity input model}) can be rewritten as
\begin{equation}\label{input velocity model}
\begin{split}
    d{p}_x&=-k_x{p}_xdt+\widetilde{{v}}_xdt+\sigma d{w}_x,\\ d{p}_y&=-k_y{p}_ydt+\widetilde{{v}}_ydt+\sigma d{w}_y.
\end{split}
\end{equation}
Now, the goal is to design an optimal control policy $u^* = [\widetilde{v}^*_x\;\widetilde{v}^*_y]^\top$ for the chance-constrained problem (\ref{CC-SOC}). The initial state is $x_0 = [ -0.3 \; 0.3 ]^\top$. We set $k_x=k_y=0.5$, $V\left({x}(t)\right) = {p}_x^2(t) + {p}_y^2(t)$, $\psi\left({x}(T)\right) = {p}_x^2(T) + {p}_y^2(T)$, $R=\begin{bmatrix}
    1 & 0\\
    0& 1
\end{bmatrix}$, $t_0=0$, $T=2$, and $\sigma^2=0.01$. {\color{black} Notice that here $\lambda=0.01$ satisfies Assumption \ref{Assum: linearity}.} 
{\color{black}In Algorithm \ref{alg:cap}, we set the error tolerance $\epsilon = 0.01$ and learning rate $\gamma=0.01$}. For FDM, the computational domain is discretized by a grid of $96\times 96$ points and the solver ode45 is used with a relative error tolerance equal to $10^{-3}$. {\color{black}In the path integral simulation, for Monte Carlo sampling, $10^5$ trajectories and a step size equal to $0.01$ are used.} 

\begin{figure}
    \centering
      \begin{tabular}{c c}
\!\!\!\!\!\!\!\!\includegraphics[scale=0.35]{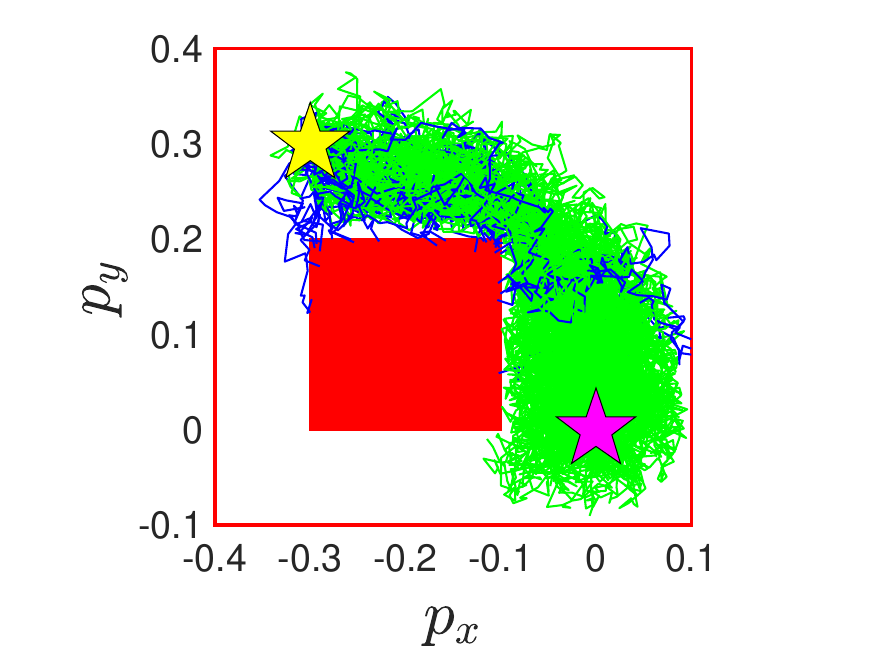} &\!\!\!\!\!\!\!\!\!\!\!\!\!\!\!\!\!\!\includegraphics[scale=0.35]{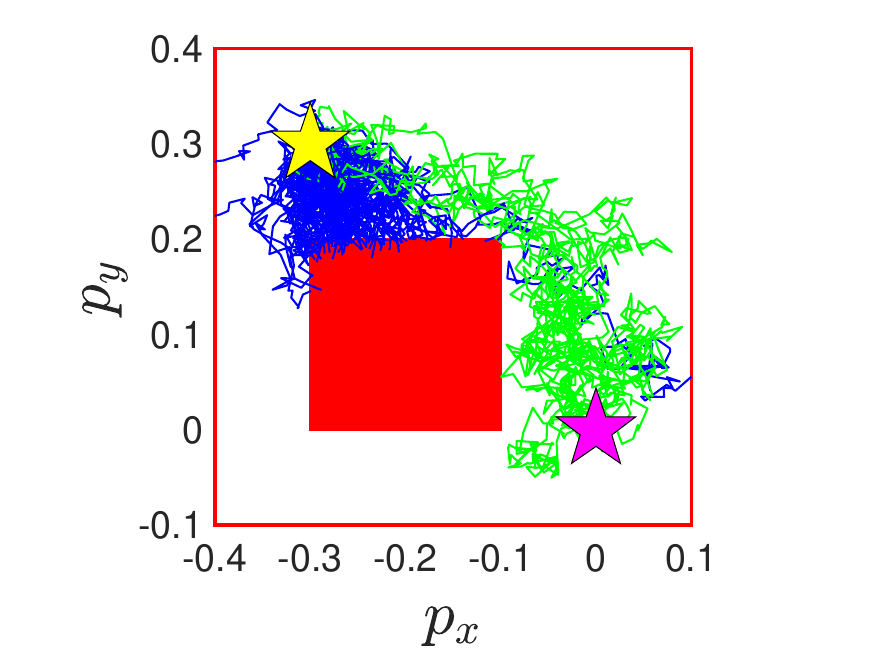} \\
      \!\!\!\!\!\!\!\!(a) $\Delta = 0.1$ & \!\!\!\!\!\!\!\!\!\!\!\!\!\!\!\!\!\! (b) $\Delta = 0.9$ \\
      \end{tabular}
        \caption{Robot navigation problem for the input velocity model. The start position is shown by a yellow star and the target position (the origin) by a magenta star. $100$ sample trajectories generated using optimal control policies for two values of $\Delta$ are shown. The trajectories are color-coded; blue paths collide with the obstacle or the outer boundary, while the green paths converge in the neighborhood of the magenta star.} 
        \label{Fig. sample trajs 2D}
\end{figure}

In Figure \ref{Fig. sample trajs 2D}, we plot $100$ sample trajectories generated using synthesized optimal policies for two values of $\Delta$. The trajectories are color-coded; the blue paths collide with the obstacle or the outer boundary, while the green paths converge in the neighborhood of the origin (the goal position). For a lower value of $\Delta$ the weight of blue paths is less and that of the green paths is more as compared to the higher value of $\Delta$. In other words, the failure probability for the lower value of $\Delta$ is less as compared to that of the higher value of $\Delta$.
\begin{figure}
    \centering
      \begin{tabular}{c c}
     \!\!\!\!\!\!\!\!\!\!\!\!\!\!\!\includegraphics[scale=0.125]{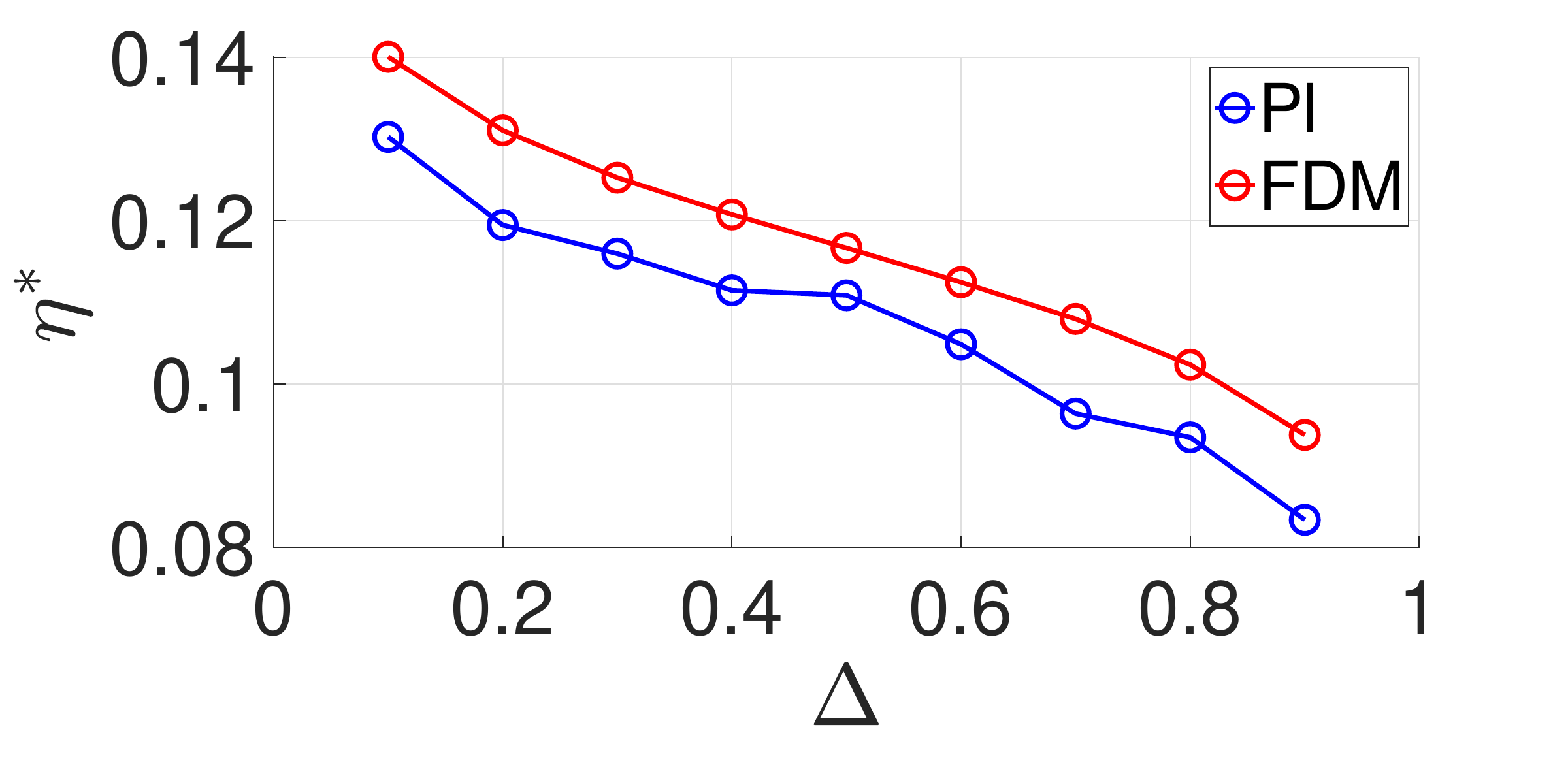} &\!\!\!\!\!\!\!\!\!\!\!\!\!\includegraphics[scale=0.125]{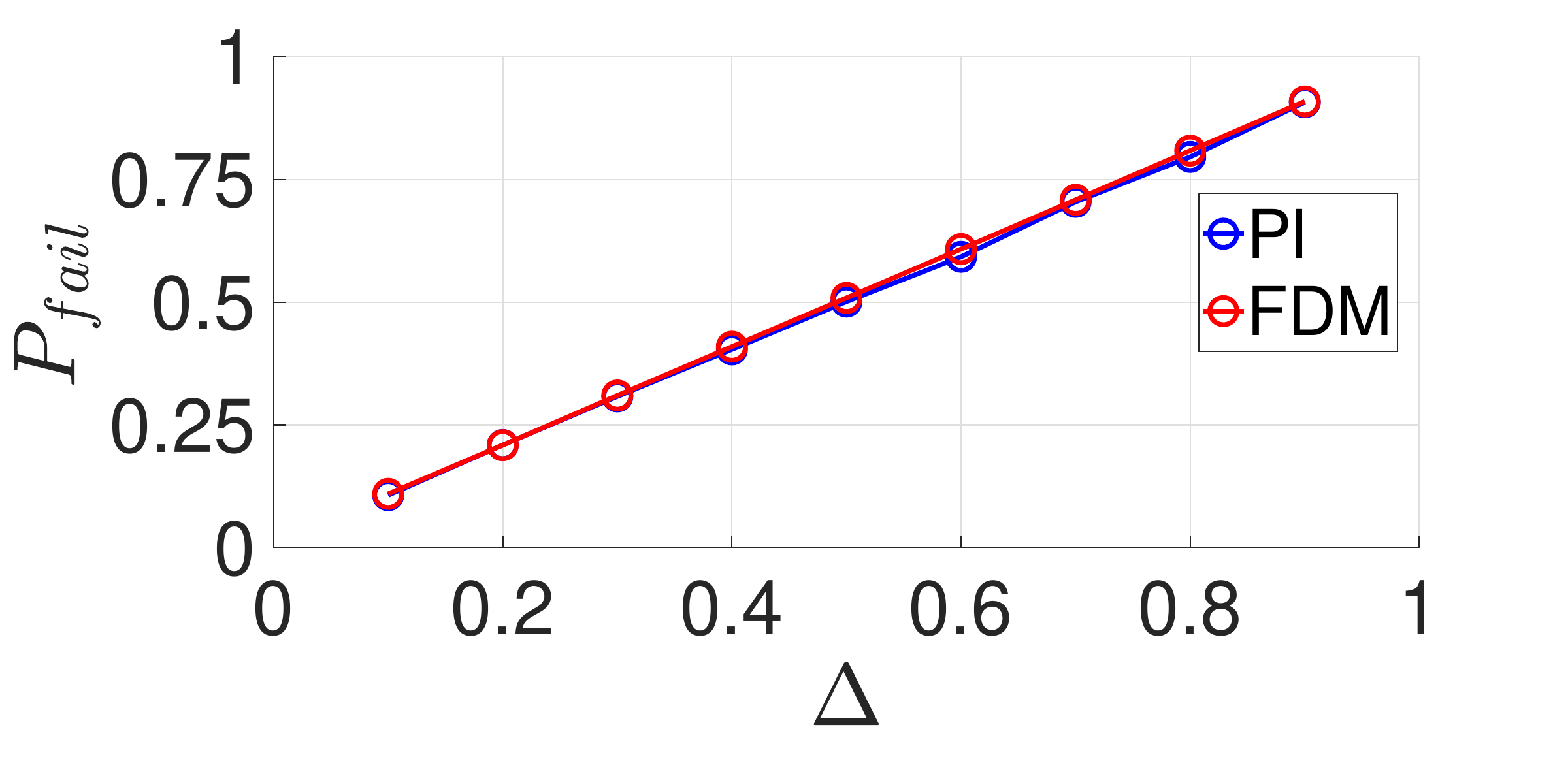} \\
      \!\!\!\!\!\!\!\!\!\!\!\!\!\!\!(a) $\eta^*$ vs $\Delta$ & \!\!\!\!\!\!\!\!\!\!\!\!\!(b) $P_{\mathrm{fail}}(x_0,t_0, u^*(\cdot\;;{\eta^*}))$ vs $\Delta$ \\
      \end{tabular}
        \caption{$\eta^*$ and $P_{\mathrm{fail}}(x_0,t_0, u^*(\cdot\;;{\eta^*}))$ vs $\Delta$ for input velocity model using path integral control and FDM.} 
        \label{Fig.eta pfail vs Delta(2D)}
\end{figure}

Figure \ref{Fig.eta pfail vs Delta(2D)} represents how the value of $\eta^*$ and $P_{\mathrm{fail}}(x_0,t_0, u^*(\cdot\;;{\eta^*}))$ change with respect to $\Delta$. The values obtained using path integral control and FDM are compared. As we expect the value of $\eta^*$ reduces and that of $P_{\mathrm{fail}}$ increases as $\Delta$ increases. 
% Note that the graphs are not monotonic. This is because we used numerical methods to solve the PDE \eqref{linearized risk-minimizing HJB}, which caused some errors in the computation of $\eta^*$ and $P_{\mathrm{fail}}(x_0,t_0, u^*(\cdot\;;{\eta^*}))$.

\begin{figure*}[h]
    \centering
      \begin{tabular}{c c c c}
      \!\!\!\!\!\!\!\!\!\!\includegraphics[scale=0.3]{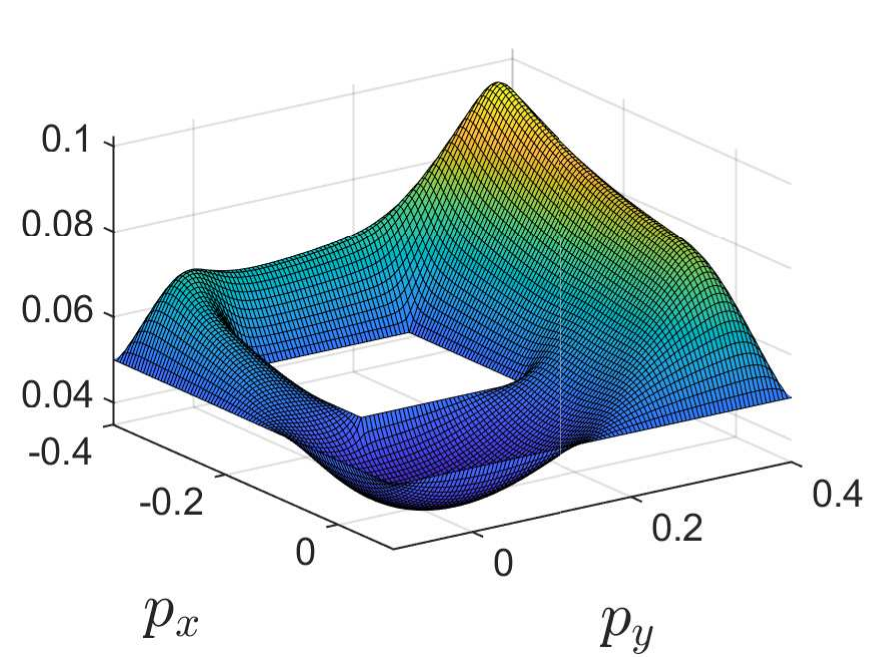} &\!\!\!\!\!\!\!\!\includegraphics[scale=0.3]{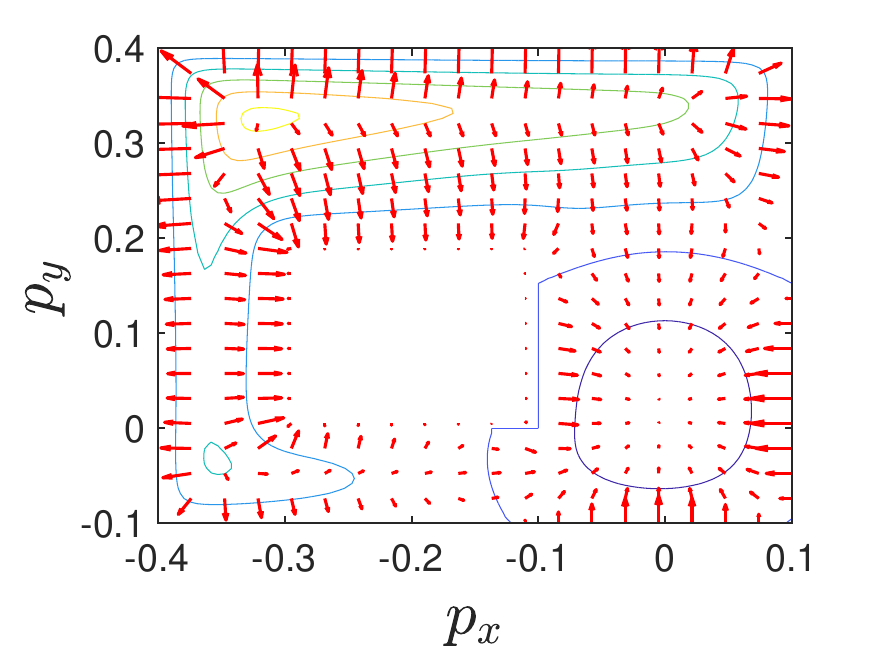} &\!\!\!\!\!\!\!\!\includegraphics[scale=0.3]{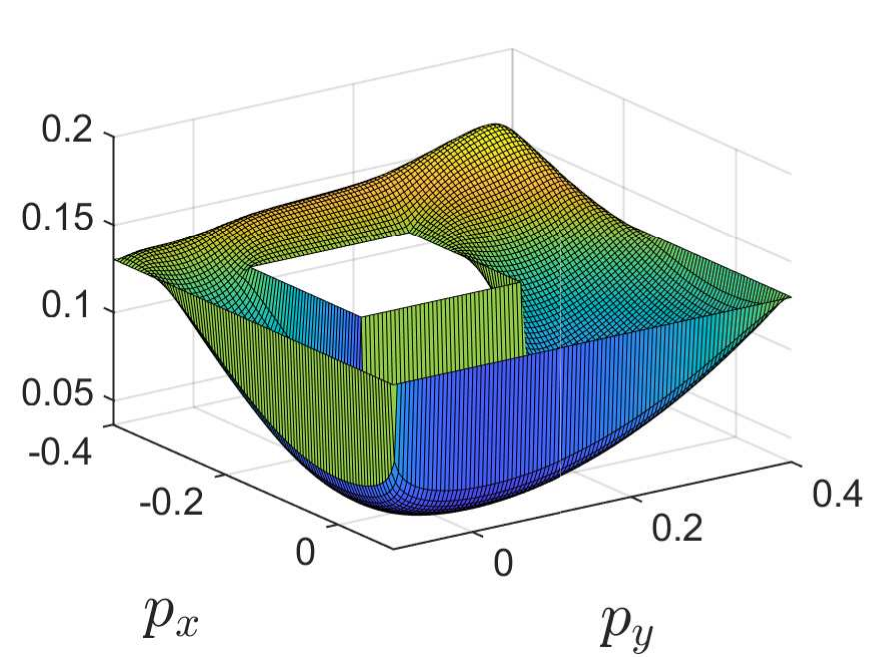} &\!\!\!\!\!\!\!\!\!\! \includegraphics[scale=0.3]{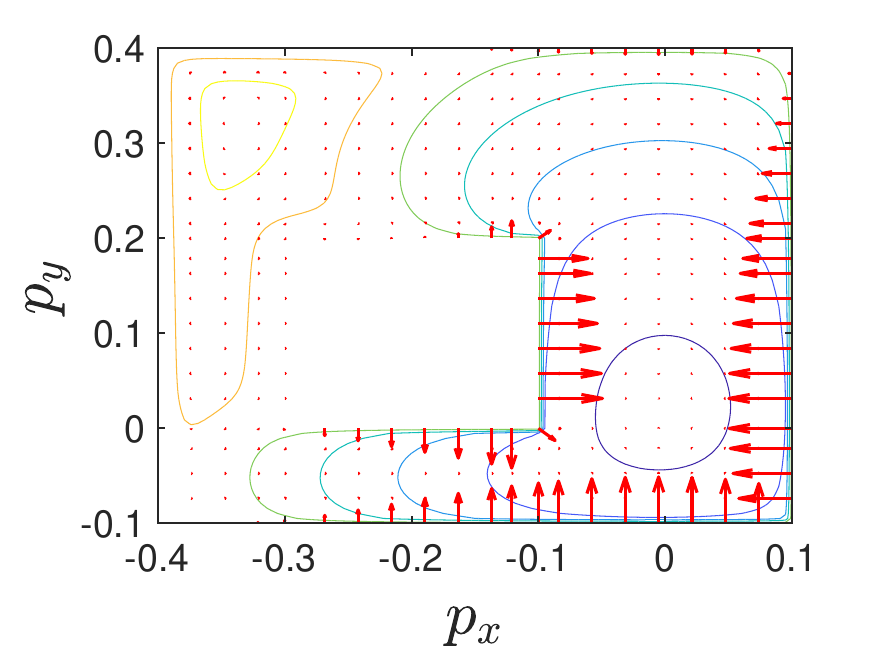}\\
      (a) $J(x,t_0; \eta)$ for $\eta = 0.05$ & (b) $u^*(x,t_0; \eta)$ for $\eta = 0.05$ & (c) $J(x,t_0; \eta)$ for $\eta = 0.13$ & (d) $u^*(x,t_0; \eta)$ for $\eta = 0.13$\\
      
        \!\!\!\!\!\!\!\!\!\!\includegraphics[scale=0.3]{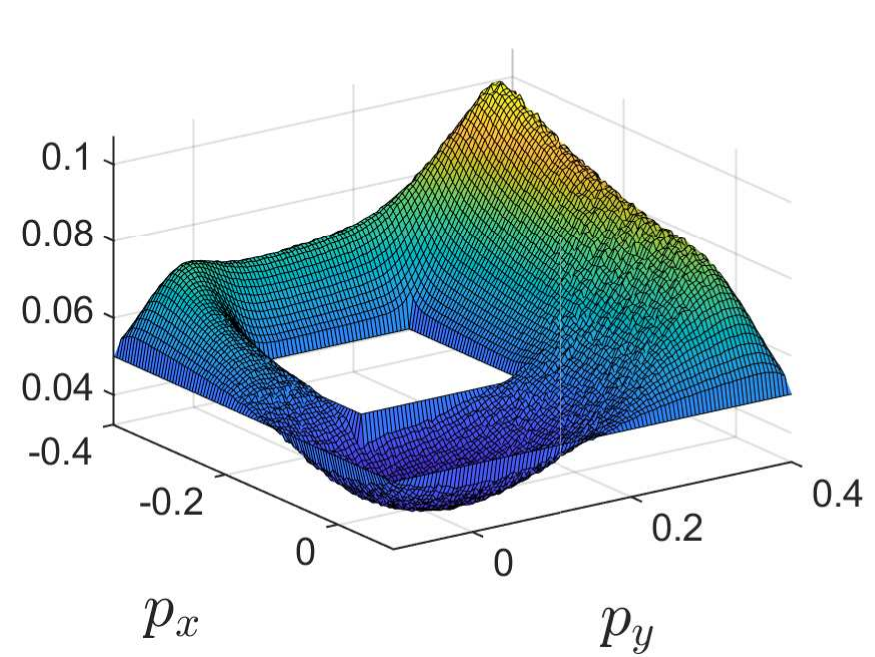} &\!\!\!\!\!\!\!\!\includegraphics[scale=0.3]{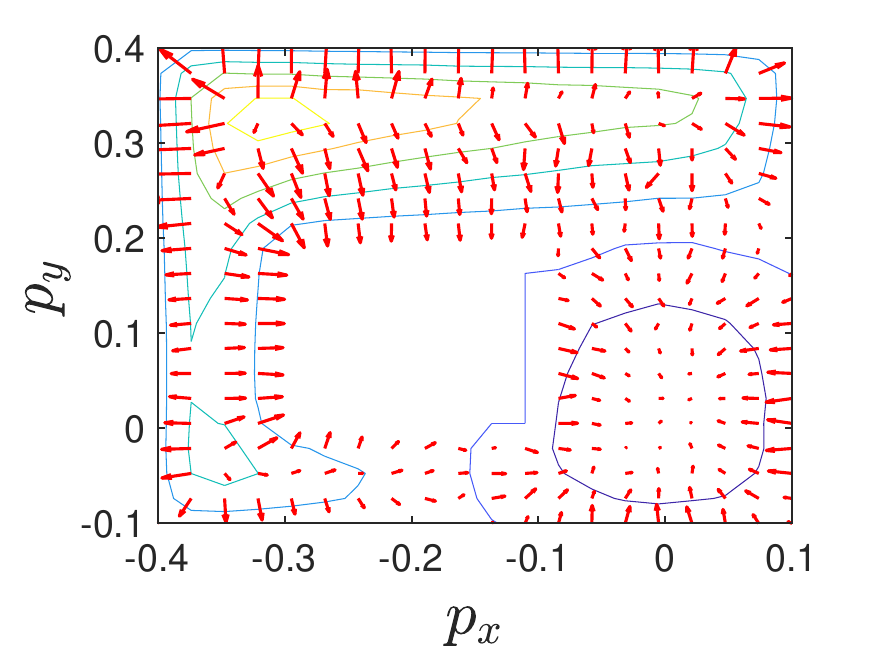}&\!\!\!\!\!\!\!\!\includegraphics[scale=0.3]{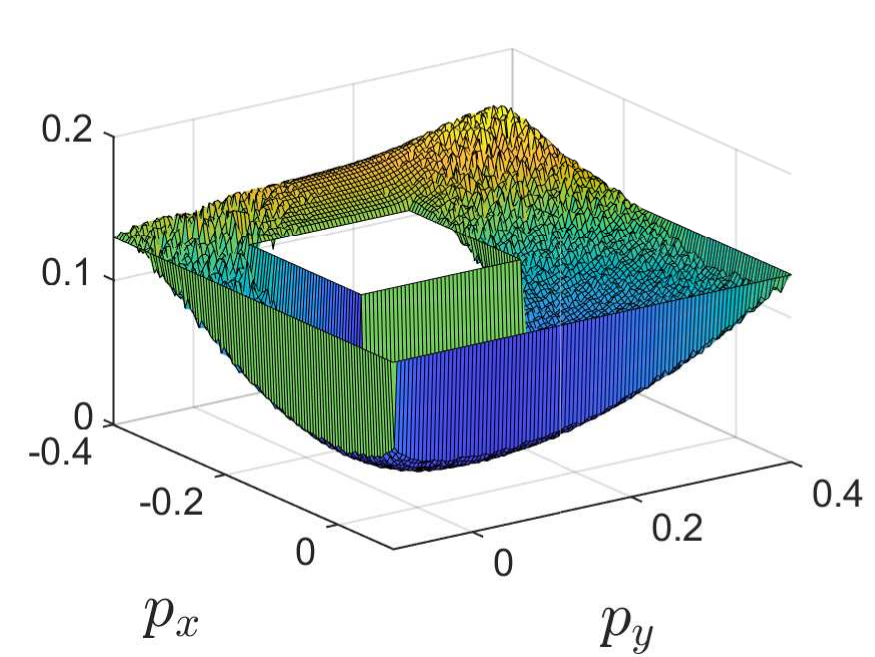} &\!\!\!\!\!\!\!\!\!\!\includegraphics[scale=0.3]{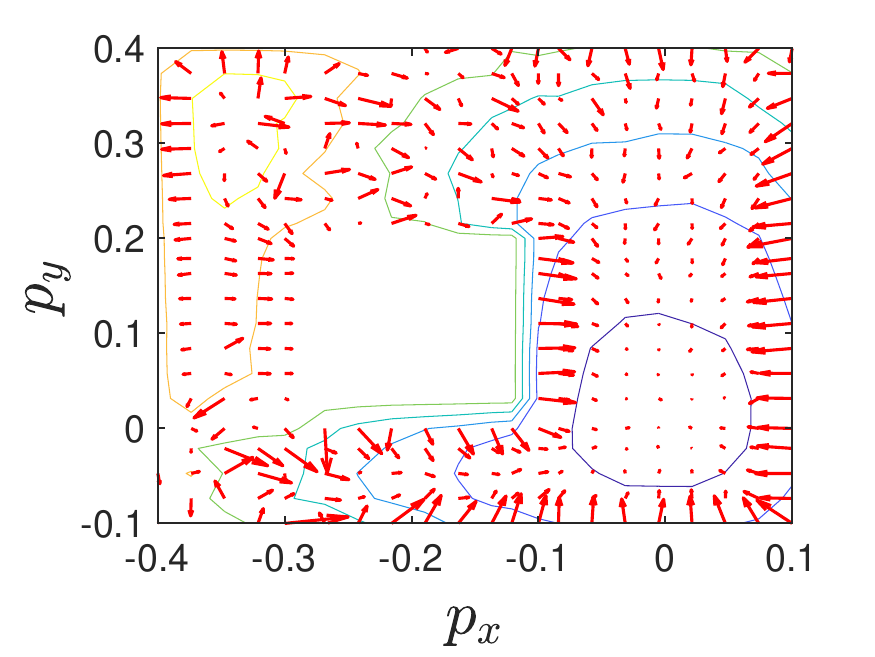}\\
        (e) $J(x,t_0;\eta)$ for $\eta = 0.05$ & (f) $u^*(x,t_0;\eta)$ for $\eta = 0.05$ & (g) $J(x,t_0; \eta)$ for $\eta = 0.13$ & (h) $u^*(x,t_0; \eta)$ for $\eta = 0.13$\\
      \end{tabular}
        \caption{Input velocity model: comparison of solutions $J(x,t_0)$ and $u^*(x,t_0)$ obtained from FDM (a-d) and path integral (e-h) for $\eta=0.05$ and $\eta=0.13$. The optimal control inputs $u^*(x,t_0)$ in (b, d, f, h) are plotted together with contours of $J(x, t_0)$.} 
        \label{Fig. surf plots FDM}
\end{figure*}
Figure \ref{Fig. surf plots FDM} shows the comparison of solutions $J(x,t_0)$ and $u^*(x,t_0)$ of the PDE \eqref{HJB PDE} obtained from FDM (top row) and path integral (bottom row) for $\eta=0.05$ and $\eta=0.13$. Since the path integral is a sampling-based approach its solutions are noisier compared to FDM, as expected. Notice that for $\eta=0.05$, due to a lower boundary value, the control inputs $u^*(x,t_0)$ push the robot towards the obstacle or the outer boundary from the most part of the safe region $\mathcal{X}_{s}$, except near the origin (the target position). Whereas, for $\eta=0.13$, aggressive inputs $u^*(x,t_0)$ are applied near the boundary to force the robot to go towards the goal position.

Figure \ref{Fig. risks} shows the colormaps of the failure probabilities of the synthesized optimal policies for $\Delta=0.1$ and $\Delta=0.9$ as functions of initial position. Notice that the region of higher failure probabilities is more for $\Delta = 0.9$ than that for the lower value of $\Delta = 0.1$.\par

    The factors that affect the computation speed of FDM include the grid size, the choice of an ODE solver and its error tolerances. Whereas the computation speed of path integral depends on the number of Monte Carlo samples. The computation time is evaluated on a machine with an Intel Core i7-9750H CPU clocked at 2.6 GHz. Both FDM and path integral approaches are implemented in MATLAB. Path integral simulations are run parallelly using MATLAB's parallel processing toolbox and the GPU Nvidia GeForce GTX 1650. {\color{black}Note that FDM solves the  HJB PDE (\ref{linearized risk-minimizing HJB}) and computes the optimal control inputs at all the state-time $(x,t)$ pairs in one go. Whereas the path integral controller computes the optimal control input at a specific state-time pair by running Monte Carlo simulations. For the given value of $\eta$, the average computation time for FDM to obtain the optimal control inputs across all state-time pairs is approximately $10.28$ seconds. In the path integral framework, the average computation time to compute the optimal control input at each time step is $0.04$ seconds.} Note that no special effort was made to optimize the algorithm for speed. We plan to reduce the computation time of the algorithm in future work.
    
\subsection{Car Model}\label{Sec: sim_5D}
The problem is illustrated in Figure \ref{Fig. sample trajs 5D}. A car wants to navigate in a 2D space from a given start position (shown by a yellow star) to the origin (shown by a magenta star), by avoiding the collisions with the red obstacles and the outer boundary. The states of the car model $\begin{bmatrix}
    {p}_x&
    {p}_y&
    {s}&
    {\theta}&
    {\phi}
    \end{bmatrix}^\top$ consists of its $x-y$ position $\begin{bmatrix}
    {p}_x&
    {p}_y
    \end{bmatrix}^\top$, the speed $s$, heading angle $\theta$ and the front wheel angle $\phi$. $L$ is the inter-axle distance. The system dynamics are given by the following SDE:
    
    \begin{equation} \label{car model}
\begin{aligned}
    \begin{bmatrix}
    d{p}_x\\d{p}_y\\d{s}\\d{\theta}\\d\phi
    \end{bmatrix}\!=&\!
    -k
    \begin{bmatrix}
    {p}_x\\
    {p}_y\\
    {s}\\
    {0}\\
    {0}
    \end{bmatrix}dt+
    \begin{bmatrix}
    {s}\cos{{\theta}}\\{s}\sin{{\theta}}\\0\\\frac{s\tan{\phi}}{L}
    \\0
    \end{bmatrix}\!dt\!\\
    &\!\!\!\!+\begin{bmatrix}
    0 & 0\\0 & 0\\1 & 0\\0 & 0\\0 & 1
    \end{bmatrix} \!
    \left(\begin{bmatrix}
    a\\
    \zeta
    \end{bmatrix}\!dt \!+\!
    \begin{bmatrix}
    \sigma & 0\\
    0 & \nu 
    \end{bmatrix}\!d{w}
    \right).
\end{aligned}
\end{equation}
    The control input $u=\begin{bmatrix} a & \zeta \end{bmatrix}^\top$ consists of acceleration $a$ and the front wheel angular rate $\zeta$. $\sigma$ and $\nu$ are the noise level parameters. In the simulation, we set $\sigma=\nu=0.07$, $k=0.2$, $t_0=0$, $T=10$, $L=0.05$, $x_0=\begin{bmatrix}
-0.4 &-0.4 & 0 &0& 0
\end{bmatrix}^\top$, $V({x}) = {p}_x^2 + {p}_y^2$, $\psi\left({x}(T)\right) = {p}_x^2(T) + {p}_y^2(T)$ and $R = \begin{bmatrix}
    1 & 0 \\ 0 & 1 
\end{bmatrix}$. {\color{black} Notice that here $\lambda=0.005$ satisfies Assumption \ref{Assum: linearity}.} We solve the chance-constrained problem \eqref{CC-SOC} via path integral approach using Algorithm \ref{alg:cap}. {\color{black} In Algorithm \ref{alg:cap}, we set the error tolerance $\epsilon=0.01$ and the learning rate $\gamma = 0.001$.} Since the state-space of the system is of high order we do not use FDM to solve this problem. For Monte Carlo sampling, $10^5$ trajectories and a step size equal to $0.01$ are used. \par
In Figure \ref{Fig. sample trajs 5D}, we plot $100$ sample trajectories generated using synthesized optimal policies for two values of $\Delta$. The trajectories are color-coded similar to the problem in Section \ref{Sec: sim_2D}. We can observe that the failure probability for the lower value of $\Delta$ is less as compared to that of the higher value of $\Delta$.
Figure \ref{Fig.eta pfail vs Delta(5D)} represents how the value of $\eta^*$ and $P_{\mathrm{fail}}(x_0,t_0, u^*(\cdot\;;{\eta^*}))$ change with respect to $\Delta$. As expected the value of $\eta^*$ reduces and that of $P_{\mathrm{fail}}$ increases as $\Delta$ increases.\par

The algorithm is implemented on the same machine mentioned in Section \ref{Sec: sim_2D}. It is run in MATLAB, and the  Monte Carlo simulations are run parallelly using MATLAB’s parallel processing toolbox. 
% The average number of iterations required by the proposed algorithm to converge is $11$. We choose the initial $\eta$ to be $0.77$ and the learning rate $\gamma$ to be $0.1$. 
 {\color{black} For the given value of $\eta$, the average computation time required to compute the optimal control input at every time step is $0.48$ seconds.} Note that no special effort was made to optimize the algorithm for speed. We plan to reduce the computation time of the algorithm in future work.

\begin{figure}
    \centering
      \begin{tabular}{c c}
     \!\!\!\!\!\!\!\!\!\includegraphics[scale=0.33]{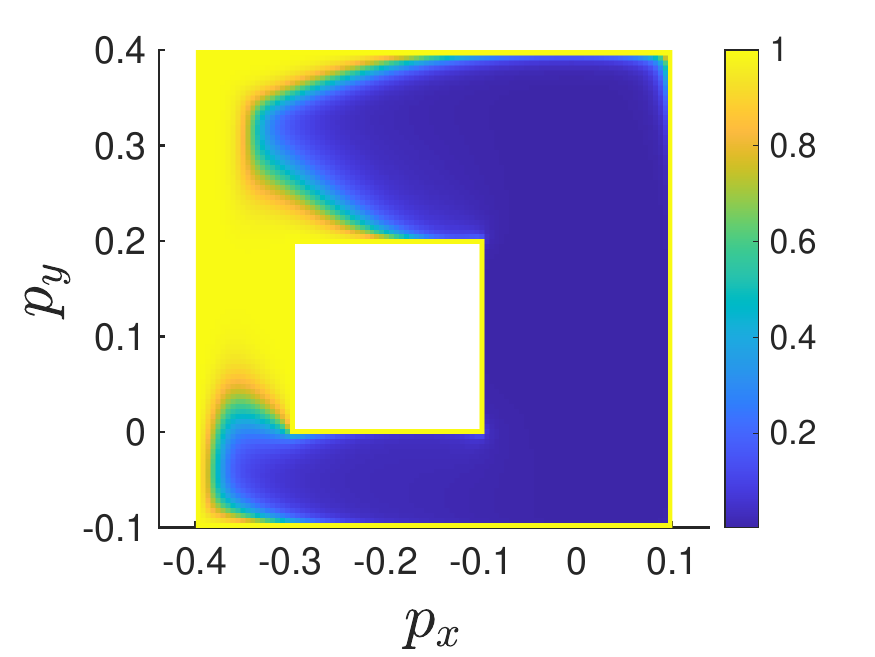} &\!\!\!\!\!\!\!\!\!\!\!\!\!\!\!\!\!\!\!\!\!\!\!\includegraphics[scale=0.33]{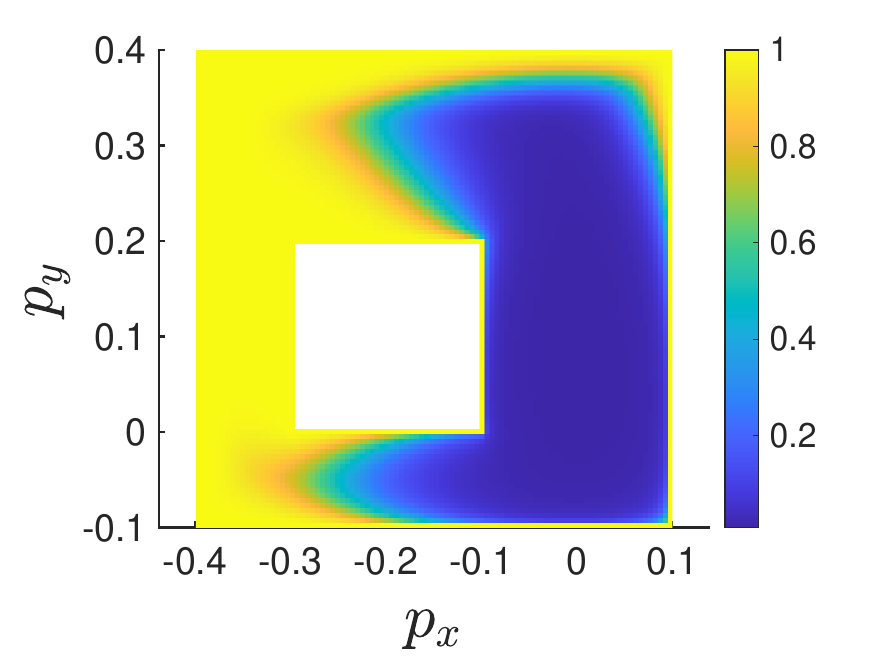} \\
      (a) $\Delta = 0.1$ & \!\!\!\!\!\!\!\!(b) $\Delta = 0.9$
      \end{tabular}
        \caption{Colormaps of the failure probabilities of the optimal policies synthesized for the input velocity model as functions of initial position.} 
        \label{Fig. risks}
\end{figure}

\begin{figure}
    \centering
      \begin{tabular}{c c}
     \!\!\!\!\!\!\!\!\!\!\!\!\!\!\!\!\includegraphics[scale=0.4]{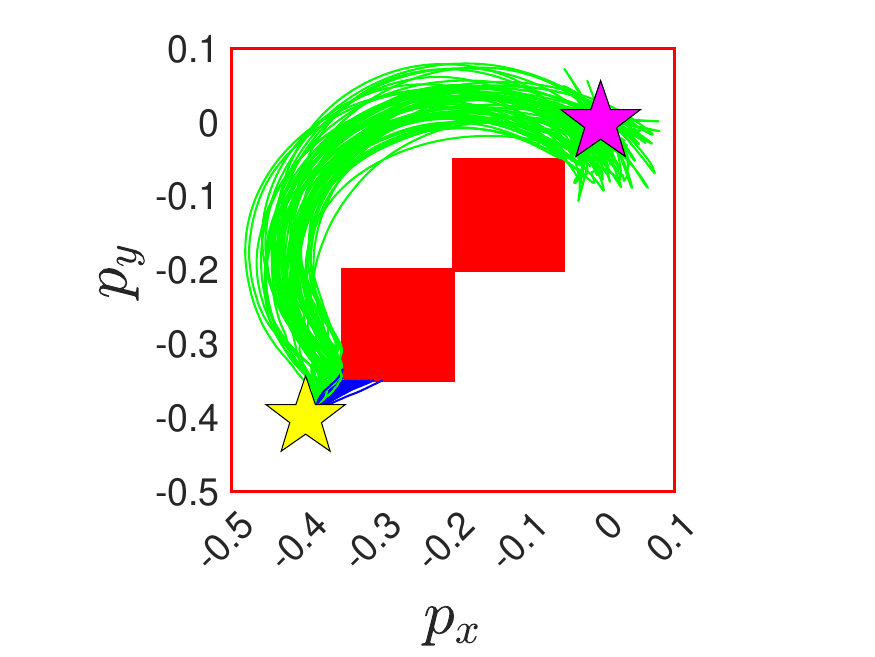} &\!\!\!\!\!\!\!\!\!\!\!\!\!\!\!\!\!\!\!\!\!\!\!\!\!\!\!\!\includegraphics[scale=0.4]{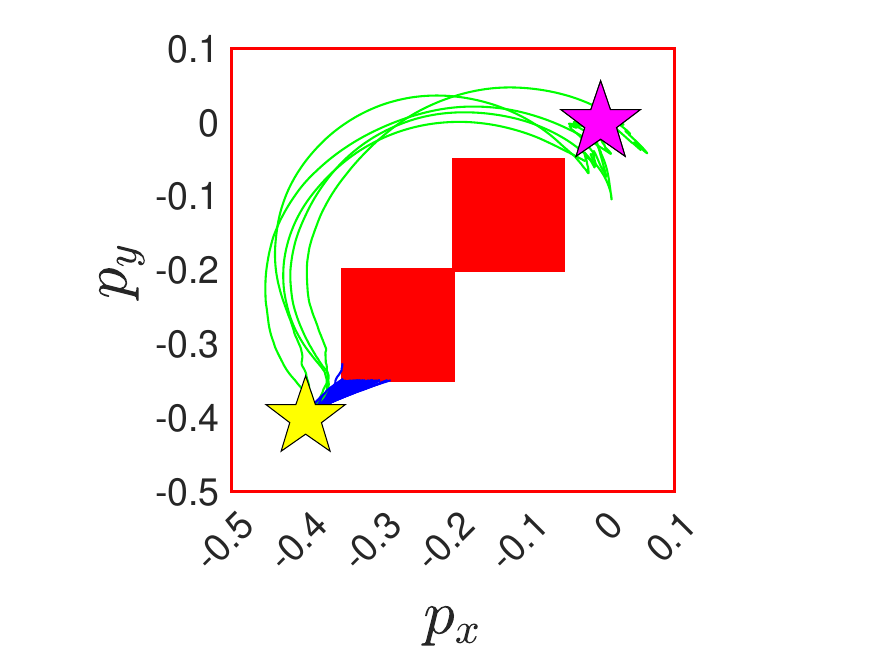} \\
      \!\!\!\!\!\!\!\!\!\!\!\!\!\!\!\! (a) $\Delta = 0.1$ & \!\!\!\!\!\!\!\!\!\!\!\!\!\!\!\!\!\!\!\!\!\!\!\!\!\!\!\! (b) $\Delta = 0.8$\\
      \end{tabular}
        \caption{Robot navigation problem for a car model. The start position is shown by a yellow star and the target position (the origin) by a magenta star. $100$ sample trajectories generated using optimal control policies for two values of $\Delta$ are shown. The trajectories are color-coded; blue paths collide with the obstacle or the outer boundary, while the green paths converge in the neighborhood of the magenta star.} 
        \label{Fig. sample trajs 5D}
\end{figure}

\begin{figure}
    \centering
      \begin{tabular}{c c}
     \!\!\!\!\!\!\!\!\!\includegraphics[scale=0.095]{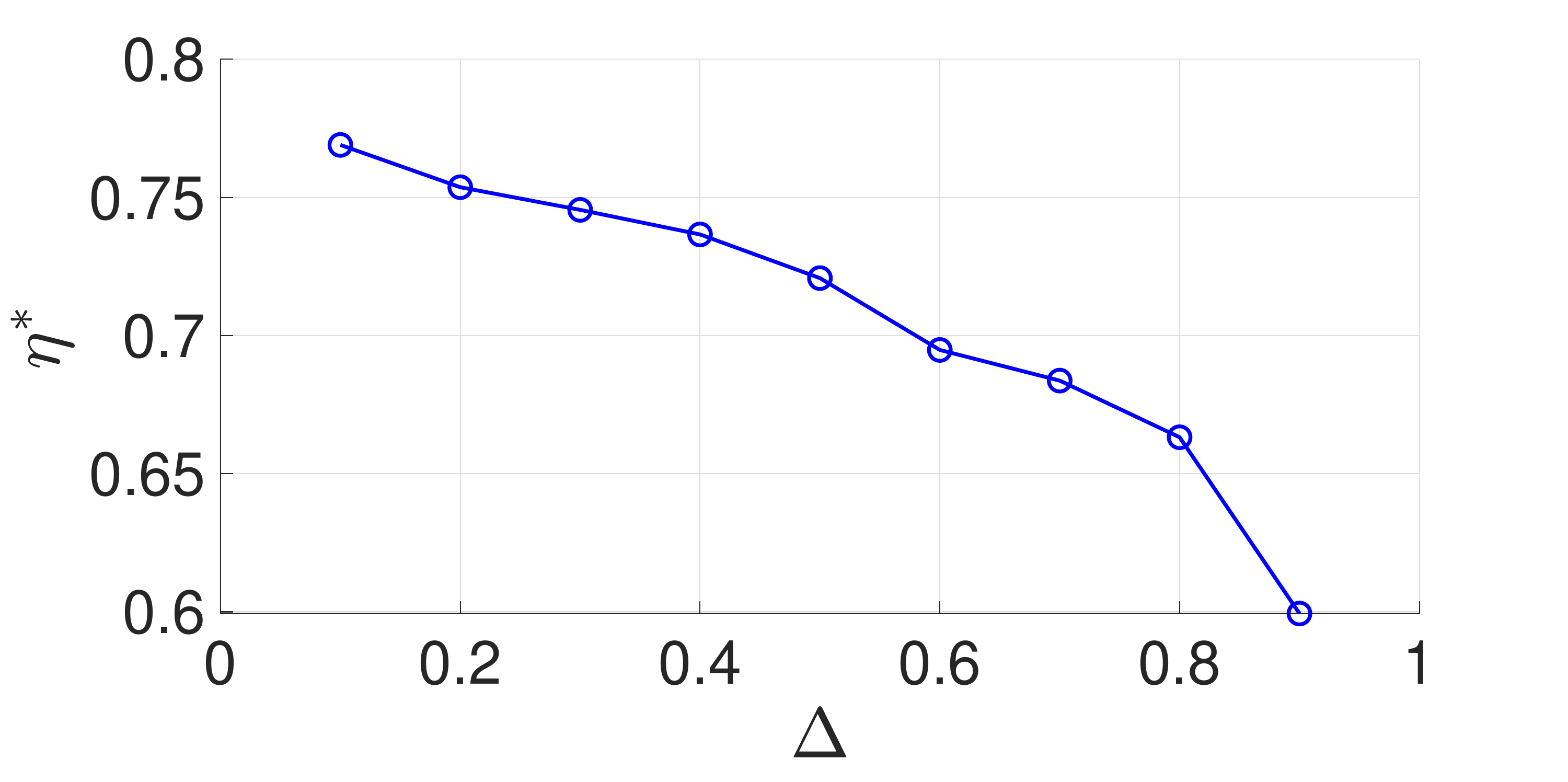} &\!\!\!\!\!\!\!\!\!\!\!\!\!\includegraphics[scale=0.095]{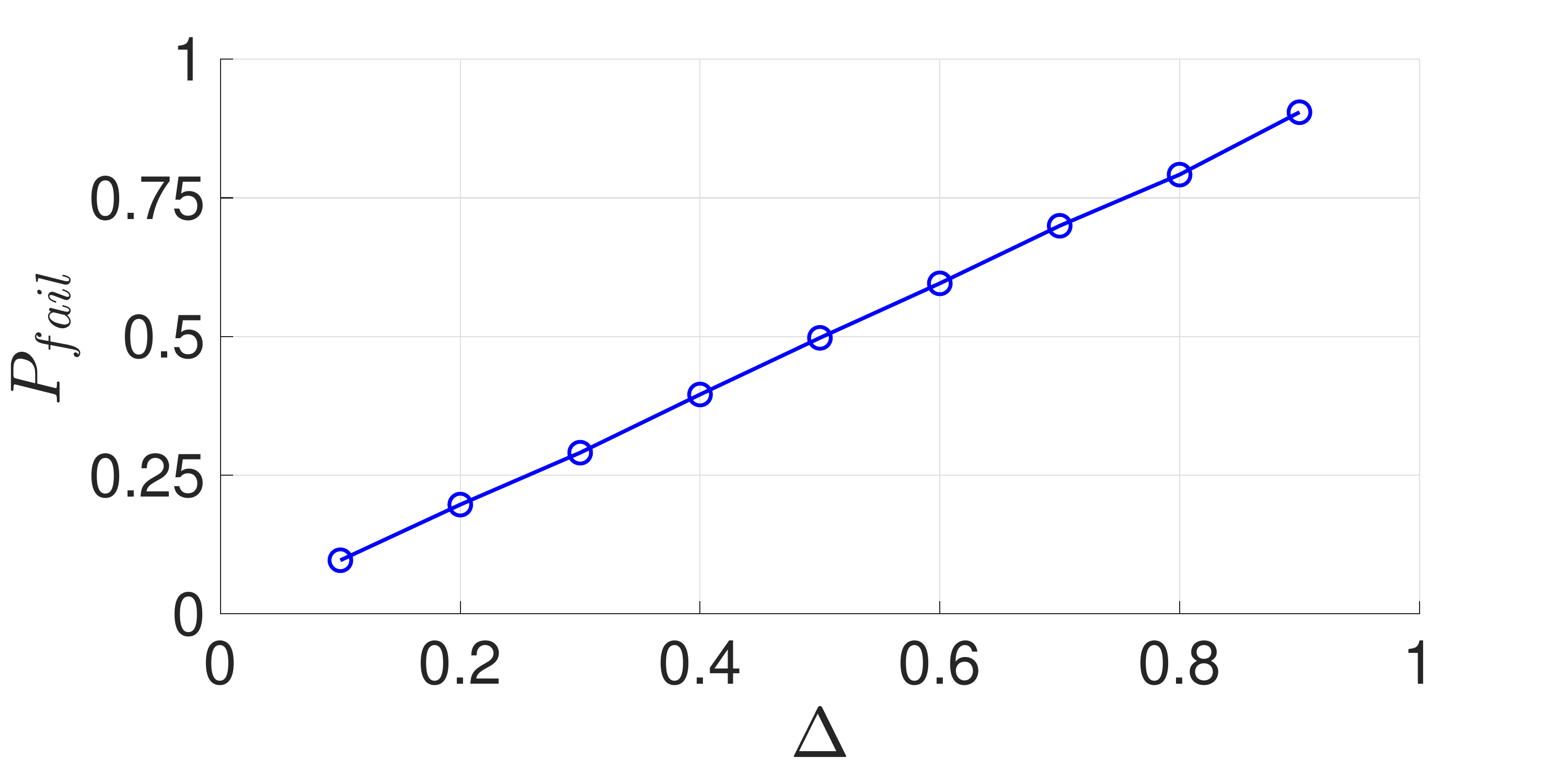} \\
      (a) $\eta^*$ vs $\Delta$ & \!\!\!\!\!\!\!\!\!\!\!\!\!(b) $P_{\mathrm{fail}}(x_0,t_0, u^*(\cdot\;;{\eta^*}))$ vs $\Delta$ \\
      \end{tabular}
        \caption{$\eta^*$ and $P_{\mathrm{fail}}(x_0,t_0, u^*(\cdot\;;{\eta^*}))$ vs $\Delta$ for a car model using path integral control.} 
        \label{Fig.eta pfail vs Delta(5D)}
\end{figure}
\section{Conclusion}
In this work, we numerically solved a continuous-time continuous-space chance-constrained stochastic optimal control (SOC) problem where the probability of failure to satisfy given state constraints is explicitly bounded. We formulated a dual SOC problem and evaluated the dual objective function by minimizing the Lagrangian via the Hamilton-Jacobi-Bellman (HJB) partial differential equation (PDE). Under a certain assumption on the system dynamics and cost function, it was shown that a strong duality holds between the primal problem and the dual problem. We proposed a novel path-integral-based dual ascent algorithm to numerically solve the dual problem. This allowed us to solve the original chance-constrained problem online via open-loop samples of system trajectories. We presented three simulation studies using an input velocity model (2D state-space system), unicycle model (4D state-space system), and car model (5D state-space system) to demonstrate our control synthesis framework. For the input velocity model we compared the solution obtained using path integral control with the finite difference method.\par

In order to prove the strong duality between the primal chance-constrained problem \eqref{CC-SOC} and its dual \eqref{eq: dual problem}, we required Assumption \ref{assum: continuity of Pfail}, where we assumed continuity of the function $\eta\mapsto P_{\mathrm{fail}}(x_0,t_0, u^*(x,t;\eta)):[0, \infty)\rightarrow[0,1]$. In the future, we plan to conduct a formal analysis which will prove that the Assumption \ref{assum: continuity of Pfail} holds true under mild conditions. We also plan to reduce the computation time of the proposed algorithm. Apart from that we plan to conduct the sample complexity analysis of the path integral control in order to investigate how the accuracy of Monte Carlo sampling affects the solution of chance-constrained SOC problems. Some preliminary results on this topic are presented in \cite{patil2024discrete}. In order to achieve a strong duality between the primal chance-constrained problem and its dual, in this paper, we require a certain assumption on the system dynamics and the cost function (Assumption \ref{Assum: linearity}). This assumption restricts the class of applicable system models and cost functions. In future work, we plan to find alternatives in order to get rid of this restrictive assumption (one such solution is provided in \cite{satoh2016iterative}). Another topic of future investigation is chance-constrained stochastic games. 
% \section*{Acknowledgments}. 

{\appendix}

\subsection{Nonconvex Optimization and Strong Duality}\label{appe: Nonconvex Optimization and Strong Duality}
Let $\mathcal{X}$ be a vector space. Suppose $f_0$ and $f_1$ are not necessarily convex, real-valued functions defined on $\mathcal{X}$. 
The domains of the functions $f_0$ and $f_1$ are denoted by $\text{dom}(f_0)$ and $\text{dom}(f_1)$, respectively.
Consider the following optimization problem with a single inequality constraint:
\begin{subequations}
\label{eq:primal_prob}
\begin{align}
\min_{x\in \mathcal{X}} &\quad f_0(x) \\
&\quad f_1(x)\leq 0.
\end{align}
\end{subequations}
We assume $\text{dom}(f_0)\cap\text{dom}(f_1)$ is nonempty, and $f_0(x)>-\infty \;\; \forall x \in \text{dom}(f_0)\cap\text{dom}(f_1)$.

Let $\eta\geq 0$ be the dual variable. Define the dual function $g:[0,\infty)\rightarrow \bar{\mathbb{R}}$ by
\begin{equation}
\label{eq:def_dual_func}
g(\eta)=\inf_{x\in\mathcal{X}} \; f_0(x)+\eta f_1(x).
\end{equation}
Since \eqref{eq:def_dual_func} is a pointwise infimum of a family of affine functions, it is concave. Moreover, since affine functions are upper semicontinuous, $g(\eta)$ is also upper semicontinuous.

\begin{assumption}
\label{asmp:strict_feasible}
    Problem \eqref{eq:primal_prob} is strictly feasible. That is, there exists $x_0\in \text{dom}(f_0)$ such that $f_1(x_0)<0$.
\end{assumption}

\begin{lemma}
\label{lem:existence_dual_sol}
    Under Assumption~\ref{asmp:strict_feasible}, there exists a dual optimal solution $\eta^*$ such that $0\leq \eta^*<\infty$ such that $g(\eta^*)=\sup_{\eta \geq 0} g(\eta)$.
\end{lemma}
\begin{proof}
Since $g(\eta)$ is upper semicontinuous, by Weierstrass' theorem \cite[Proposition A.8]{bertsekas1999nonlinear}, it is sufficient to prove that the function $-g(\eta)$ is coersive. Specifically, it is sufficient to show that there exists a scalar $\gamma$ such that the set 
\[
\eta(\gamma):=\{\eta \geq 0 \mid g(\eta) \geq \gamma\}
\]
is nonempty and compact.

We will show that the choice $\gamma:=\inf_{x\in\mathcal{X}} f_0(x)$ will make $\eta(\gamma)$ nonempty and bounded. To see that $\eta(\gamma)$ is nonempty, notice that 
\[
g(0)=\inf_{x\in\mathcal{X}} f_0(x)=\gamma
\]
and thus $0\in \eta(\gamma)$. 

To prove $\eta(\gamma)$ is bounded, we need to show that $g(\eta)\searrow -\infty$ as $\eta \rightarrow +\infty$. Let $x_0$ be a strictly feasible solution, i.e., $f_1(x_0)<0$. We have
\begin{align*}
g(\eta)&=\inf_x \; f_0(x)+\eta f_1(x) \\
&\leq f_0(x_0)+\eta \underbrace{f_1(x_0)}_{<0}\rightarrow -\infty 
\end{align*}
as $\eta \rightarrow +\infty$. This completes the proof.
\end{proof}
To proceed further, we need an additional set of assumptions.

\begin{assumption}
\label{asmp:continuity}
    \begin{itemize}
    \item[(i)] For each $\eta\geq 0$, the set
    \[
    X(\eta):=\argmin_{x\in\mathcal{X}} \; f_0(x)+\eta f_1(x)
    \]
    is nonempty.
    \item[(ii)] The function $\eta \mapsto f_1(X(\eta)): [0,\infty)\rightarrow \mathbb{R}$ is well-defined. That is, if $x_1\in X(\eta)$ and $x_2\in X(\eta)$, then $f_1(x_1)=f_1(x_2)$.
    \item[(iii)] The function $f_1(X(\cdot))$ is continuous.
    \end{itemize}
\end{assumption}
{\color{black}Assumption \ref{asmp:continuity} serves as part of the constraint qualification required to establish strong duality. Together with Assumption 4, Assumption 5 ensures the technical conditions necessary to prove complementary slackness in Lemma \ref{lem:comp_slackness}.}

\begin{lemma}
\label{lem:comp_slackness}
Suppose Assumption~\ref{asmp:strict_feasible} holds so that a dual optimal solution $0\leq \eta^*<+\infty$ exists (c.f., Lemma~\ref{lem:existence_dual_sol}). {\color{black}In addition, suppose Assumption \ref{asmp:continuity} also holds true,} then the following statements hold:
\begin{itemize}
\item[(i)] If $\eta^*=0$, then $f_1(X(\eta^*))\leq 0$.
\item[(ii)] If $\eta^*>0$, then $f_1(X(\eta^*))= 0$.
\end{itemize}
\end{lemma}
\begin{proof}
(i) Suppose $\eta^*=0$ but $f_1(X(\eta^*))>0$. Set $\eta_\epsilon=\epsilon>0$.
By continuity of $f_1(X(\cdot))$, we have $f_1(X(\eta_\epsilon))>0$ for a sufficiently small $\epsilon$. Notice that
\begin{align*}
g(\eta_\epsilon)&=f_0(X(\eta_\epsilon))+\underbrace{\eta_\epsilon f_1(X(\eta_\epsilon))}_{>0} \\
&>f_0(X(\eta_\epsilon)) \\
&\geq \min_{x\in\mathcal{X}} \; f_0(x) \\
&=\min_{x\in\mathcal{X}}  \; f_0(x)+\eta^* f_1(x) \\
&=g(\eta^*).
\end{align*}
However, this chain of inequalities contradicts to the fact that $\eta^*$ is the dual optimal solution (i.e., $\eta^*$ maximizes $g(\eta)$).

(ii) Suppose $\eta^*>0$ and $f_1(X(\eta^*))>0$.
Set $\eta_\epsilon=\eta^*+\epsilon$ where $\epsilon>0$ is sufficiently small.
By continuity of $f_1(X(\cdot))$, we have $f_1(X(\eta_\epsilon))>0$. Notice that 
\begin{align*}
g(\eta_\epsilon)&=f_0(X(\eta_\epsilon))+\eta_\epsilon f_1(X(\eta_\epsilon))\\
&=f_0(X(\eta_\epsilon))+\eta^* f_1(X(\eta_\epsilon))+\underbrace{\epsilon f_1(X(\eta_\epsilon))}_{>0} \\
&>f_0(X(\eta_\epsilon))+\eta^* f_1(X(\eta_\epsilon)) \\
&\geq \min_{x\in\mathcal{X}}  \; f_0(x)+\eta^* f_1(x) \\
&=g(\eta^*).
\end{align*}
This contradicts to the fact that $\eta^*$ maximizes $g(\eta)$.

Similarly, suppose $\eta^*>0$ and $f_1(X(\eta^*))<0$.
Set $\eta_\epsilon=\eta^*-\epsilon$ and choose $\epsilon>0$ sufficiently small so that $\eta_\epsilon>0$ and $f_1(X(\eta_\epsilon))<0$. We have
\begin{align*}
g(\eta_\epsilon)&=f_0(X(\eta_\epsilon))+\eta_\epsilon f_1(X(\eta_\epsilon))\\
&=f_0(X(\eta_\epsilon))+\eta^* f_1(X(\eta_\epsilon))-\epsilon \underbrace{f_1(X(\eta_\epsilon))}_{<0} \\
&>f_0(X(\eta_\epsilon))+\eta^* f_1(X(\eta_\epsilon)) \\
&\geq \min_{x\in\mathcal{X}}  \; f_0(x)+\eta^* f_1(x) \\
&=g(\eta^*).
\end{align*}
Again, this contradicts to the fact that $\eta^*$ maximizes $g(\eta)$. Therefore, if $\eta^*>0$ then $f_1(X(\eta^*))=0$ must hold.
\end{proof}
The following is the main result of this section.
\begin{theorem}\label{theorem: strong duality2}
    Consider problem \eqref{eq:primal_prob} and suppose Assumptions~\ref{asmp:strict_feasible} and \ref{asmp:continuity} hold. Then, there exists a dual optimal solution $0\leq \eta^*<+\infty$ that maximizes $g(\eta)$.
    Moreover, the set \[
    X(\eta^*):=\argmin_{x\in\mathcal{X}} \; f_0(x)+\eta^* f_1(x)
    \]
    is nonempty, and any element $x^*\in X(\eta^*)$ is a primal optimal solution such that $f_0(x^*)=g(\eta^*)$. i.e., the duality gap is zero. Consequently, a minimizer of $f_0(x)+\eta^*f_1(x)$ is an optimal solution of \eqref{eq:primal_prob}.
\end{theorem}
\begin{proof}
By Lemma~\ref{lem:comp_slackness}, any element $x^*\in X(\eta^*)$ satisfies $f_1(x^*)\leq 0$. Hence, $x^*$ satisfies primal feasibility.
It also follows from Lemma~\ref{lem:comp_slackness} that $\eta^* f_1(x^*)=0$, i.e., the complementary slackness condition holds. Therefore,
\begin{align*}
f_0(x^*)&=f_0(x^*)+\eta^* f_1(x^*) \quad \text{(Complementary slackness)}\\
&=\min_{x\in \mathcal{X}} f_0(x)+\eta^* f(x) \quad \text{(since $x^*\in X(\eta^*)$)} \\
&=g(\eta^*).
\end{align*}
Hence, the strong duality holds.
Therefore, $x^*$ is a primal optimal solution.
It follows from the identity above that a minimizer of $f_0(x)+\eta^*f_1(x)$ is an optimal solution for \eqref{eq:primal_prob}.
\end{proof}

\subsection{Legendre Duality}
\label{appendix:a}

Let $P$ and $Q$ be probability distributions on $(\mathcal{X}, \mathcal{B}(\mathcal{X}))$, and $C:\mathcal{X}\rightarrow \mathbb{R}$ a given cost function. Define the internal energy $U(Q,C)$, free energy $F(P,C)$ and relative entropy (KL divergence) $D(Q\|P)$ as:
\begin{align*}
&U(Q,C):=\int_\mathcal{X} C(x)Q(dx)  \\
&F(P,C):= -\lambda \log \int_\mathcal{X} \exp\left(-\frac{C(x)}{\lambda}\right)P(dx) \\
&D(Q\|P):= \int_\mathcal{X} \log\frac{dQ}{dP}(x)Q(dx).
\end{align*}
Then the following duality relationship holds:
\begin{align*}
F(P,C)&=\inf_Q \{ U(Q,C)+\lambda D(Q\|P) \} \\
    -\lambda D(Q\|P)&=\inf_C \{ U(Q,C)-F(P,C) \}.
\end{align*}
Also, the optimal probability distribution $Q^*$ is given by
\begin{equation*}
    Q^*(B) = \frac{\int_{B}\exp(-C(x)/\lambda)P(dx)}{\int_{\mathcal{X}} \exp(-C(x)/\lambda)P(dx)}, \quad \forall B \in \mathcal{B}(\mathcal{X}).
\end{equation*}
See \cite{boue1998variational,theodorou2012relative, patil2023simulator} for further discussions.

\subsection{The Likelihood Ratio $\frac{dQ^*}{dP}(x)$}
Let $Q(x)$ be the probability distribution of the trajectories defined by system \eqref{SDE} under the given policy $u(\cdot)$ and $Q^*(x)$ be the probability distribution of the trajectories defined by system \eqref{SDE} under the optimal policy $u^*(\cdot)$ Let $P(x)$ be the probability distribution of the trajectories defined by the uncontrolled system \eqref{uncontrolled SDE}. Suppose the likelihood ratio of observing a sample path $x$ under the distributions $Q$ and $P$ is denoted by $\frac{dQ}{dP}(x)$ and the expectation under any distribution $Q$ is denoted by $\mathbb{E}^Q[\cdot]$. Let $k(t)$ be a process defined by
{\color{black}
\begin{equation}\label{k(t)}
    k(t) \coloneqq \Sigma^{\dagger}Gu(t)
\end{equation}
where $\Sigma^{\dagger}$ is the left pseudo-inverse of \( \Sigma \) i.e. it holds \eqref{sigma_dagger sigma}.} Now, we state the following theorem:
\begin{theorem}[The Girsanov theorem]
    The likelihood ratio $\frac{dQ}{dP}(x)$ of observing a sample path $x$ under distributions $Q$ and $P$ is given by 
    \begin{equation}\label{dQ*/dP}
        \frac{dQ}{dP}(x) = \exp\left({\int_{t_0}^{{{t}}_{f}}}k(t)^\top d{w}(t) + \frac{1}{2}\int_{t_0}^{{{t}}_{f}}\|k(t)\|^2dt\right)
    \end{equation}
    where the process $k(t)$ is defined by \eqref{k(t)}. 
\end{theorem}
\begin{proof}
    Refer to \cite[Chapter 8.6]{oksendal2013stochastic}, \cite{patil2025path}.
\end{proof}

\begin{theorem}\label{Thm: likelihood ratio}
    Suppose we generate an ensemble of $N$ trajectories $ \{x^{(i)}\}_{i=1}^N$ under the distribution $P$. Let $r^{(i)}$ be the path reward associated with the trajectory $x^{(i)}$ as given in \eqref{r(i)}. Define $r \coloneqq \sum_{i=1}^{N}r^{(i)}$. Then as $N\rightarrow\infty$, 
    
    \begin{equation*}
    \frac{r^{(i)}}{r/N} \overset{a.s.}{\rightarrow} \frac{dQ^*}{dP}(x)     
    \end{equation*}
\end{theorem}
\begin{proof}
    The likelihood ratio $\frac{dQ}{dP}(x)$ of observing a sample path $x$ is given by \eqref{dQ*/dP}. Using this result, we obtain
   
    \begin{equation}\label{E log likelihood}
        \mathbb{E}^{Q}\log\left(\frac{dQ}{dP}(x)\right) =\frac{1}{2}\int_{t_0}^{{t}_f} \mathbb{E}^Q\|k(t)\|^2dt.
    \end{equation}
In \eqref{E log likelihood}, we used the property of It\^o integral \cite[Chapter 3]{oksendal2013stochastic}:
    \begin{equation*}
        \mathbb{E}^Q\left[{\int_{t_0}^{{{t}}_{f}}}k(t)^\top d{w}(t)\right] = 0
    \end{equation*} 
    Now, for a given value of $\eta$, we wanted to solve the following problem
    \begin{align}\label{optimization problem}
        \min_{u(\cdot)}\mathbb{E}^Q_{x_0, t_0}\!\!\left[\phi\!\left({x}({t}_{f}); \eta\right)\!+\!\!\!\int_{t_0}^{{t}_{f}}\!\!\!\left(\!\frac{1}{2}{u}^\top\!R{u}+V\!\right)\!dt\!\right].
    \end{align}
 {\color{black}  Multiplying both sides of Assumption~\eqref{lambda} from the left by \( \Sigma^{\dagger} \) and from the right by \( (\Sigma^{\dagger})^\top \), we obtain:
\begin{equation}
    \Sigma^{\dagger} \Sigma \Sigma^\top (\Sigma^{\dagger})^\top =  \lambda \Sigma^{\dagger} G R^{-1} G^\top (\Sigma^{\dagger})^\top 
% I = &\left(\lambda^{\frac{1}{2}} \Sigma^{\dagger} G R^{-\frac{1} {2}}\right)\left(\lambda^{\frac{1}{2}}R^{-\frac{1}{2}}G^\top (\Sigma^{\dagger})^\top \right)\nonumber
\end{equation}
Since $\Sigma^{\dagger}\Sigma=I$ by \eqref{sigma_dagger sigma}, it follows that
\begin{equation*}
  \lambda \Sigma^{\dagger} G R^{-1} G^\top (\Sigma^{\dagger})^\top = I . 
\end{equation*}
Defining
\begin{equation*}
 M=\lambda^{\frac{1}{2}} \Sigma^{\dagger} G R^{-\frac{1}{2}},   
\end{equation*}
we get $MM^\top = I$. Note that matrix $M\in \mathbb{R}^{k\times m}$. However, without loss of generality, we may assume $m=k$, that is, $M$ is a square matrix. The justification is as follows: Assumption \ref{Assum: linearity} (equation \eqref{lambda}) implies that the left-hand side has rank $k$, since $\Sigma$ is assumed to have full column rank. However, the rank of the right-hand side of \eqref{lambda} is at most $m$, which implies that $m\geq k$. Moreover, since the left-hand side has rank $k$, the rank of $G\in\mathbb{R}^{n\times m}$ must be $k$. That means only $k$ directions in control input space actually affect the dynamics. Consequently, under Assumption~\ref{Assum: linearity}, the original chance-constrained stochastic optimal control (SOC) problem (Problem~\ref{Problem: Risk-constrained SOC problem}) with an $m$-dimensional control input can always be reformulated as an equivalent problem with a $k$-dimensional control input. Therefore, assuming that $M$ is a square matrix, we also have $ M^\top M = I$. This yields the identity:
\begin{equation}\label{R intermediate}
    \lambda R^{-\frac{1}{2}}G^\top(\Sigma^\dagger)^\top\Sigma^\dagger G R^{-\frac{1}{2}} = I.
\end{equation}
Multiplying both sides of \eqref{R intermediate} from left and right by $R^{\frac{1}{2}}$, we get  

    \begin{equation}\label{linearity 2}
    R = \lambda G^\top(\Sigma^\dagger)^{\top}\Sigma^{\dagger}G.
    \end{equation}}
   Now, we obtain
\begin{subequations}
    \begin{align}
        &\min_{u(\cdot)}\mathbb{E}^Q_{x_0, t_0}\!\!\left[\phi\!\left({x}({t}_{f}); \eta\right)\!+\!\!\!\int_{t_0}^{{t}_{f}}\!\!\!\left(\!\frac{1}{2}{u}^\top\!R{u}+V\!\right)\!dt\!\right]\label{with R}\\
        = & \min_{u(\cdot)}\mathbb{E}^Q_{x_0, t_0}\!\!\left[\phi\!\left({x}({t}_{f}); \eta\right)\!+\!\!\!\int_{t_0}^{{t}_{f}}\!\!\!\left(\!\frac{1}{2}{u}^{\!\!\top}\!\lambda G^\top(\Sigma^\dagger)^{\top}\Sigma^{\dagger}G{u}\!+\!\!V\!\!\right)\!dt\!\right]\label{replacing R}\\
        = & \min_{k(\cdot)}\mathbb{E}^Q_{x_0, t_0}\!\!\left[\phi\!\left({x}({t}_{f}); \eta\right)\!+\!\!\!\int_{t_0}^{{t}_{f}}\!\!\!\left(\!\frac{\lambda}{2}\|k(t)\|^2+V\!\right)\!dt\!\right]\label{replacing with k}\\
        = & \min_{Q(x)}\int_{\mathcal{T}}\left( \phi\!\left({x}({t}_{f}); \eta\right) + \int_{t_0}^{{t}_f}V dt + \lambda\log\frac{dQ}{dP}(x)\right)Q(dx).\label{minimization under Q}
        \end{align}
\end{subequations}
Equation \eqref{replacing R} is obtained by plugging  \eqref{linearity 2} into \eqref{with R}. \eqref{replacing with k} obtained by using \eqref{k(t)} and \eqref{minimization under Q} by using \eqref{E log likelihood}. Thus, we converted the problem \eqref{optimization problem} into a KL control problem.
{\color{black}\begin{remark}
In reformulating the problem from \eqref{replacing with k} to \eqref{minimization under Q}, we apply \eqref{E log likelihood} that expresses the control optimization over path measures. Importantly, this change of representation does not alter the underlying requirement that the control inputs remain adapted to the system’s filtration. The probability measures $Q(x)$ in \eqref{minimization under Q} are constructed over trajectories generated by adapted control processes, ensuring that the adapted nature of the control is preserved throughout. We refer the reader to the proof of Theorem 3.1 in \cite{boue1998variational} for a formal justification.
\end{remark}}

Now invoking the Legendre duality between the KL divergence and free energy (See Appendix B) it can be shown that there exists a minimizer $Q^*$ of \eqref{minimization under Q} which can be written as 
\begin{align}\label{Q*}
    Q^*(dx) = \frac{\exp\left(-\frac{1}{\lambda}\left(\phi\!\left({x}({t}_{f}); \eta\right) + \int_{t_0}^{{t}_f}V dt\right)\right)P(dx)}{\mathbb{E}^P \left[\exp\left(-\frac{1}{\lambda}\left(\phi\!\left({x}({t}_{f}); \eta\right) + \int_{t_0}^{{t}_f}V dt\right)\right)\right]}
\end{align}
Using \eqref{Q*}, we can write the expression for the Radon-Nikodym derivative $\frac{dQ^*}{dP}(x)$ as 
\begin{equation}\label{dQ/dP}
    \frac{dQ^*}{dP}(x) = \frac{\exp\left(-\frac{1}{\lambda}\left(\phi\!\left({x}({t}_{f}); \eta\right) + \int_{t_0}^{{t}_f}V dt\right)\right)}{\mathbb{E}^P \left[\exp\left(-\frac{1}{\lambda}\left(\phi\!\left({x}({t}_{f}); \eta\right) + \int_{t_0}^{{t}_f}V dt\right)\right)\right]}
\end{equation}
Using \eqref{dQ/dP}, we can conclude that given the ensemble of N trajectories $\{x^{(i)}\}_{i=1}^N$ sampled under distribution $P$, the likelihood ratio $\frac{dQ^*}{dP}$ of observing a sample path $x^{(i)}$ is given by $\frac{r^{(i)}}{r/N}$ where $r^{(i)}$ is defined by \eqref{r(i)} and $r \coloneqq \sum_{i=1}^{N}r^{(i)}$. Using the strong law of large numbers as $N\rightarrow\infty$, we get  
\begin{equation*}
    \frac{r^{(i)}}{r/N} \overset{a.s.}{\rightarrow} \frac{dQ^*}{dP}(x).  
\end{equation*}
\end{proof}
% = \frac{\exp\left(-\frac{1}{\lambda}\left(\phi\!\left({x}({t}_{f}); \eta\right) + \int_{t_0}^{{t}_f}V dt\right)\right)

\bibliographystyle{IEEEtran}
\bibliography{references}

\begin{IEEEbiography}
[{\includegraphics[width=1in,height=1.25in,clip,keepaspectratio]{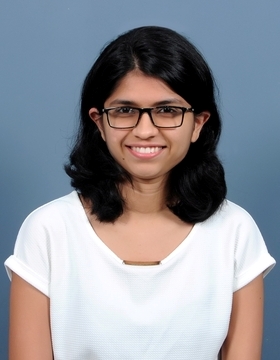}}]
{Apurva Patil} is a PhD student at UT Austin advised by Prof. Takashi Tanaka. She obtained her Bachelor's degree in Mechanical Engineering from College of Engineering Pune (COEP). At UT Austin, she received the H. Grady Rylander Excellence in Teaching Fellowship. At COEP she was recognized as a best student 2017 and received the best bachelor’s project award. During her Bachelor's she received S. N. Bose fellowship for research internship at Texas A\&M University. She is broadly interested in the intersection of control theory, robotics, and learning theory to solve problems in decision-making under uncertainty. Recent topics that she has worked on include path integral control, risk-aware motion planning, risk analysis of motion plans, and stochastic dynamic games.
\end{IEEEbiography}

\begin{IEEEbiography}
[{\includegraphics[width=1in,height=1.25in,clip,keepaspectratio]{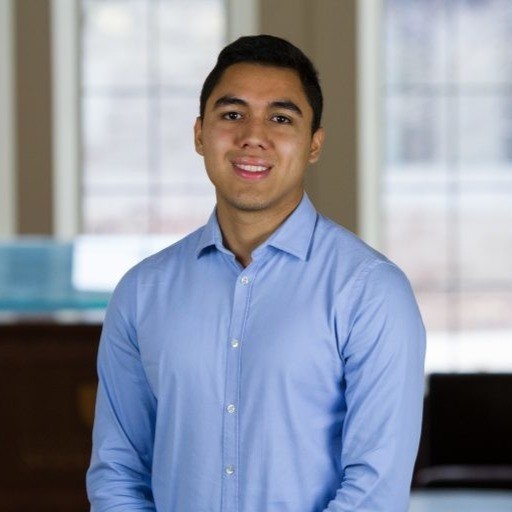}}]
{Alfredo Duarte} has been a PhD student in the University of Texas at Austin since 2019, under the supervision of Dr. Fabrizio Bisetti. He grew up in Tegucigalpa, Honduras before attending the University of Notre Dame for his bachelor's degree in Aerospace Engineering. At the University of Notre Dame he received a Naughton Fellowship for Research Experience in Ireland and graduated cum laude in 2019. His research interests are focused on the use of High Performance Computing (HPC) to study multi-temporal/multi-scale problems in engineering. Recently he has worked in the use of Jacobian-free Newton-Krylov methods to solve partial differential equations, and fluid-plasma models for the study of nanosecond discharge pulses in a pin-to-pin configuration.
\end{IEEEbiography}

\begin{IEEEbiography}
[{\includegraphics[width=1in,height=1.25in,clip,keepaspectratio]{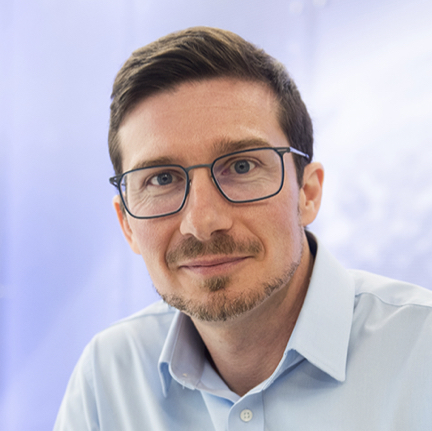}}]
{Fabrizio Bisetti} is an Assistant Professor in Aerospace Engineering at the University of Texas at Austin, where he moved in late 2016. Prior, Prof. Bisetti held a faculty appointment at King Abdullah University of Science and Technology (KAUST), where he joined the Clean Combustion Research Center (CCRC) in July 2009 as a founding faculty. Prof. Bisetti holds a Laurea (Politecnico di Milano, 6/2003), MS (UT Austin, 8/2002), and PhD (UC Berkeley, 12/2007) in Mechanical Engineering. Upon graduation, he joined the Center for Turbulence Research at Stanford University as a Postdoctoral Fellow (1/2008-6/2009). Prof. Bisetti’s research interests are in turbulent combustion, soot formation in turbulent flames, turbulent aerosols, turbulent mixing, plasma assisted ignition, and numerical methods for reactive and plasma flows. His research activities combine High Performance Computing (HPC) and theory to understand complex multi-physics/multi-scale processes in turbulent flows, e.g. aerosol/turbulence interaction, turbulent combustion, and plasma-assisted ignition.
\end{IEEEbiography}

\begin{IEEEbiography}
[{\includegraphics[width=1in,height=1.25in,clip,keepaspectratio]{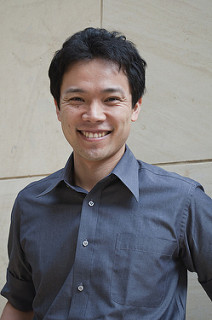}}]
{Takashi Tanaka} (M’09) received the B.S. degree from the University of Tokyo, Tokyo, Japan, in 2006, and the M.S. and Ph.D. degrees in Aerospace Engineering (automatic control) from the University of Illinois at Urbana-Champaign (UIUC), Champaign, IL, USA, in 2009 and 2012, respectively. He was a Postdoctoral Associate with the Laboratory for Information and Decision Systems (LIDS) at the Massachusetts Institute of Technology (MIT), Cambridge, MA, USA, from 2012 to 2015, and a postdoctoral researcher at KTH Royal Institute of Technology, Stockholm, Sweden, from 2015 to 2017. He is an Associate Professor in the Department of Aerospace Engineering and Engineering Mechanics at the University of Texas at Austin. Dr. Tanaka’s research interests include control theory and its applications; most recently the information-theoretic perspectives of optimal control problems. He was the recipient of the DARPA Young Faculty Award, the AFOSR Young Investigator Program award, and the NSF Career award.
\end{IEEEbiography}

\end{document}